\documentclass[10pt,aps,prx,onecolumn,superscriptaddress,floatfix,nofootinbib]{revtex4-2}
\usepackage{caption}
\usepackage[english]{babel}
\usepackage{silence}
\usepackage{amsmath}
\usepackage{amsthm}
\usepackage{amsfonts, amssymb, bbm, braket, accents}
\usepackage{graphicx}
\usepackage{wrapfig}
\usepackage[usenames,dvipsnames]{color}
\usepackage[makeroom]{cancel}
\usepackage{soul}
\usepackage{tikz} 
\usetikzlibrary{shapes.arrows} 
\usetikzlibrary{shapes.geometric}
\usetikzlibrary{calc} 
\usepackage[makeroom]{cancel}
\usepackage{verbatim}
\usepackage{pifont}
\usepackage[mathscr]{euscript}
\usepackage{amsthm}
\usepackage{enumerate}
\usepackage{ragged2e}
\def\be{\begin{eqnarray}}
\def\ee{\end{eqnarray}}
\usepackage{tikz}
\usetikzlibrary{positioning}
\usepackage{caption}
\usepackage{enumitem}
\usepackage{multirow}

\usepackage{booktabs}
\usepackage{xcolor}
\usepackage{siunitx}
\usepackage{colortbl}
\usepackage{algorithm,algpseudocode}

\newtheorem{theorem}{Theorem}
\newtheorem{proposition}[theorem]{Proposition}
\newtheorem{lemma}[theorem]{Lemma}

\newtheorem{fact}{Fact}

\theoremstyle{definition}
\newtheorem{definition}{Definition}
\newtheorem*{remark}{Remark}

\usepackage[breaklinks=true]{hyperref}
\hypersetup{
  colorlinks   = true,
  urlcolor     = blue,
  linkcolor    = blue,
  citecolor   = blue
}

\usepackage{tikz}
\usepackage{algorithm,algpseudocode}

\allowdisplaybreaks

\DeclareMathOperator{\tr}{Tr}

\usepackage[capitalise]{cleveref}
\crefformat{equation}{Eq.~(#2#1#3)}
\crefformat{section}{Sec.~#2#1#3}
\Crefformat{equation}{Equation~(#2#1#3)}
\crefformat{figure}{Fig.~#2#1#3}
\crefrangeformat{equation}{Eqs.~#3(#1)#4--#5(#2)#6}
\Crefformat{section}{Section~#2#1#3}

\begin{document}
\title{Quantum Learning with Tunable Loss Functions}

\author{Yixian Qiu}
\email{yixian\_qiu@u.nus.edu}
\affiliation{Centre for Quantum Technologies, National University of Singapore, Singapore}
\author{Lirandë Pira}
\email{lpira@nus.edu.sg}
\affiliation{Centre for Quantum Technologies, National University of Singapore, Singapore}
\author{Patrick Rebentrost}
\email{cqtfpr@nus.edu.sg}
\affiliation{Centre for Quantum Technologies, National University of Singapore, Singapore}
\affiliation{Department of Computer Science, National University of Singapore, Singapore}

\date{\today}

\begin{abstract}
Learning from quantum data presents new challenges to the paradigm of learning from data. This typically entails the use of quantum learning models to learn quantum processes that come with enough subtleties to modify the theoretical learning frameworks. This new intersection warrants new frameworks for complexity measures, including those on quantum sample complexity and generalization bounds. Empirical risk minimization (ERM) serves as the foundational framework for evaluating learning models in general. The diversity of learning problems leads to the development of advanced learning strategies such as tilted empirical risk minimization (TERM). Theoretical aspects of quantum learning under a quantum ERM framework are presented in [PRX Quantum 5, 020367 (2024)]. In this work, we propose a definition for TERM suitable to be employed when learning quantum processes, which gives rise to quantum TERM (QTERM). We show that QTERM can be viewed as a competitive alternative to implicit and explicit regularization strategies for quantum process learning. This work contributes to the existing literature on quantum and classical learning theory threefold. First, we prove QTERM learnability by deriving upper bounds on QTERM's sample complexity. Second, we establish new PAC generalization bounds on classical TERM. Third, we present QTERM agnostic learning guarantees for quantum hypothesis selection. These results contribute to the broader literature of complexity bounds on the feasibility of learning quantum processes, as well as methods for improving generalization in quantum learning.
\end{abstract}

\maketitle

\section{Introduction}
\begin{figure*}[t]
    \centering
    \includegraphics[width=0.6\linewidth]{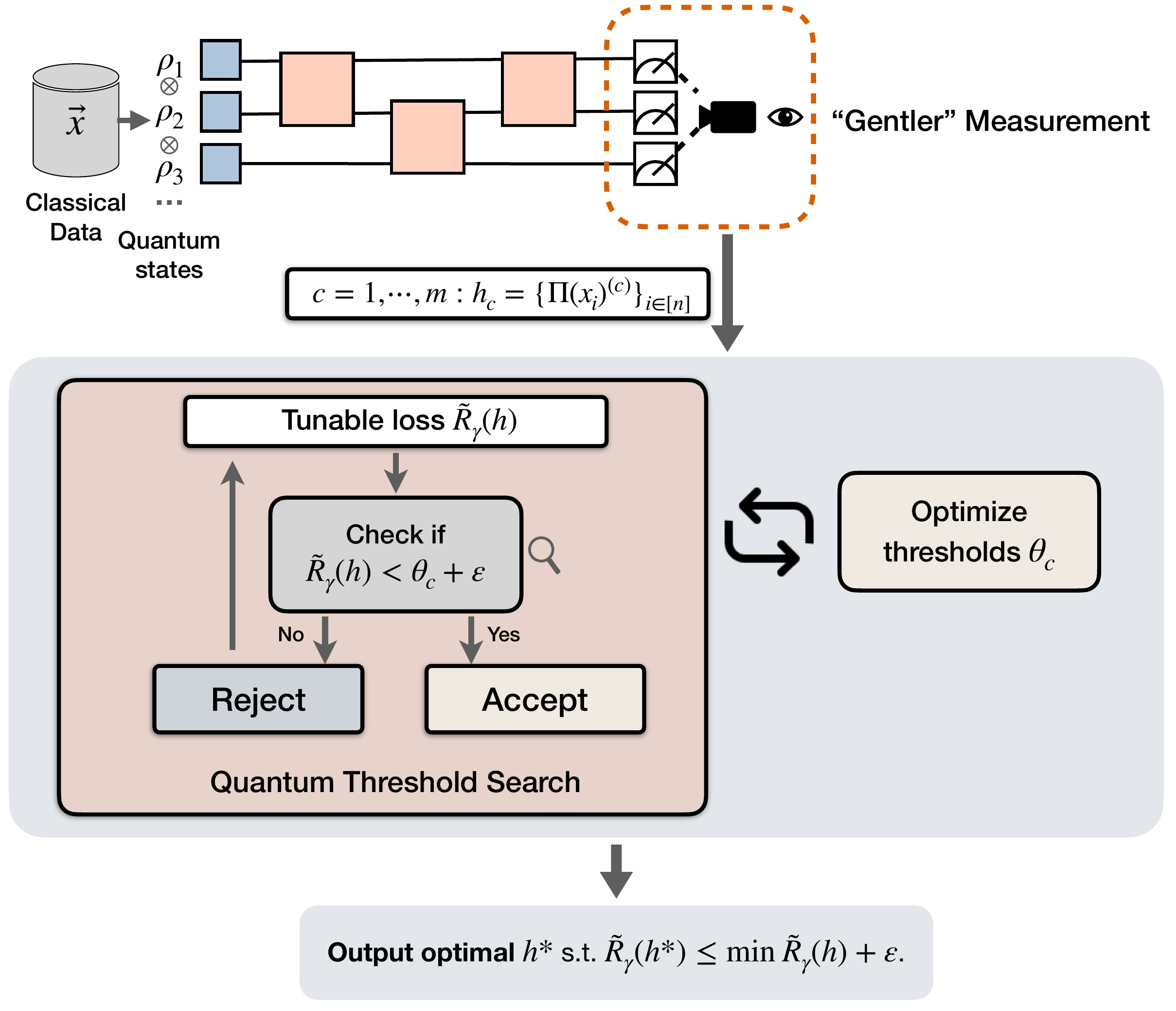}
    \caption{\justifying \textbf{Illustration of the learning process with QTERM.} The process incorporates three main parts: \textbf{Input:} It begins with classical data and the corresponding quantum state representations. \textbf{Algorithm:} It utilizes a quantum threshold search shown in Appendix.\ref{app:algo}, which aims to find an optimal hypothesis class $h_{c^*}$ that minimizes the tunable loss $\tilde{R}_\gamma(h)$, subject to the constraint $\tilde{R}_\gamma(h) \leq \theta_c$. \textbf{Output:} The optimal hypothesis set is denoted as $h_{c^*} = \{\Pi(x_i)^{(c^*)}\}_{i\in[n]}$.}
    \label{fig:process}
\end{figure*}
Learning is the cornerstone of intelligence, shaping how humans make sense of the world and adapt to it. At its core, learning is the art of discerning patterns. Humans learn by drawing connections and building relationships among objects, using intuition and experience. Similarly, learning machines seek to reveal patterns in data, but instead of intuition, they rely on mathematical models and algorithms to identify meaningful structures. Learning theory formalizes this process, providing a framework to understand how efficiently a model can generalize from limited data and determine the fundamental limits of computation in artificial learning systems \cite{valiant1984theory, kearns1994introduction, kearns_efficient_1994}.

The two primary components of machine learning are the dataset and the model \cite{bishop_pattern_2006, goodfellow_deep_2016, lecun_deep_2015}. The optimization process of learning from data can be summarized as choosing a model that best fits the dataset. Selecting such a model is a complex and often non-trivial task due to several factors. Firstly, this process is dependent on the diversity and characteristics of the data, such as its size, dimensionality, and inherent noise, which can significantly influence model performance. Next, different models have varying assumptions and strengths; for example, some may be better for linear relationships while others are designed for more complex, non-linear patterns. This is where regularization techniques, implicit or explicit, come into the picture. Regularization introduces additional constraints or penalties into the model training process, which helps prevent overfitting by discouraging overly complex models. Additionally, regularization plays an important role in preventing overfitting, handling model complexity, balancing the role of outliers and class imbalance --- all issues arising at model or dataset level. A primary objective of regularization is to reduce generalization error, which refers to the difference between a model’s performance on training data and its ability to perform well on test data.

The process of finding a model that best fits a dataset is guided by a loss function, which quantifies the discrepancy between predicted and actual outcomes. A fundamental principle in supervised learning is empirical risk minimization (ERM), which trains models by minimizing the average loss on a dataset. While ERM optimizes performance on training data, it can be refined through advanced learning techniques such as tilted empirical risk minimization (TERM) \cite{li2020tilted, li2023tilted}. TERM extends ERM by incorporating a tilted loss function that selectively emphasizes certain aspects of the dataset, leading to more nuanced and desirable learning outcomes. This approach is inspired by exponential tilting, a well-established concept in statistics \citep{siegmund1976importance, butler2007saddlepoint}, information theory \citep{thomas2006elements}, and applied probability \citep{dembo2009large}, where probability distributions are transformed by scaling them with an exponential factor. Recent studies \cite{aminian2024generalizationerrortiltedempirical} have explored the generalization properties of TERM as a competitive strategy to achieve similar outcomes as explicit regularization (additive changes to the loss function) and implicit regularization (due to stochastic gradient descent \cite{smith2021origin}), by also analyzing its theoretical implications within the broader framework of computational learning theory \cite{valiant1984theory, kearns1994introduction, kearns_efficient_1994}. In particular, TERM's performance can be evaluated using the PAC (probably approximately correct) learning framework \cite{valiant1984theory}, which assesses the accuracy of learning algorithms in relation to the amount of data required to generalize effectively. More generally, statistical theoretical bounds with respect to regularization techniques are explored in Refs.~\cite{bauer_regularization_2007, demol2008elasticnetregularizationlearningtheory}.

Quantum machine learning \cite{schuld_machine_2021, biamonte_quantum_2017, rebentrost_quantum_2014, schuld2020circuit, beer_training_2020} has emerged as a prospective application of quantum computing \cite{nielsen_quantum_2011, montanaro_quantum_2016}. Quantum machine learning algorithms in its most general form take the flavor of learning from quantum data via quantum models. Quantum models are studied and benchmarked in the near-term model of quantum computation \cite{preskill_quantum_2018} as well as the fault-tolerant one \cite{shor1996fault}. Near-term models mainly comprise the class of algorithms known as variational quantum algorithms \cite{cerezo2021variational, benedetti2019parameterized}, while in the longer-term range ideas based on quantum singular value transform (or QSVT) \cite{gilyen_quantum_2019} and block-encodings are brought forth \cite{ivashkov2024qkanquantumkolmogorovarnoldnetworks, guo_quantum_2024}. Here, we take a more holistic approach akin to the quantum model formulation under consideration, which is a generalized supervised quantum model suitable to various concept classes. The subtleties quantum characteristics add to the process of learning can be summarized on the level of data input and output, as well as the model. For learning theory measures, the quantum nature of data introduces challenges in defining generalization and overfitting due to the probabilistic nature of quantum measurements and the limitations on state distinguishability. Furthermore, measurement constraints, where only partial information about quantum states is accessible, directly affect empirical risk estimation.

In this study, we introduce quantum tilted empirical risk minimization (QTERM), which provides a formal foundation for developing quantum learning within established learning-theoretic frameworks, as depicted in \cref{fig:process}. We present QTERM as a unification of classical TERM and quantum empirical risk minimization (QERM), thereby extending ERM in the quantum setting through a generalization of both approaches. While ERM provides a foundation for model training via average loss minimization and QERM extends this to quantum data. QTERM further augments the QERM framework by introducing a tilting hyperparameter over the quantum empirical risk. This tilt hyperparameter allows selective emphasis on certain quantum samples, thereby refining the learning objective. This mechanism retains classical TERM's focus-shifting behavior of a given dataset, but now adapted to quantum measurement constraints and statistical properties of quantum states. To formalize this idea, we take an analytical approach to defining and developing QTERM for learning from quantum data. In our formulation, quantum data is represented through a collection of quantum projector-value functions, and the QERM framework \cite{fanizza_learning_2024} is extended by incorporating a tilted hyperparameter over discrete quantum samples. This formulation offers a precise operational interpretation: QTERM reduces to QERM when the tilting parameter is neutral, and to classical TERM when applied to classical data with deterministic outcomes. Using the definitions of this framework, we derive sample complexity bounds, as well as agnostic learning guarantees for QTERM in a finite hypothesis space. Furthermore, constructing an explicit PAC bound for QTERM establishes its theoretical relevance within the quantum PAC framework, while also deriving new generalization results for classical TERM under PAC assumptions. These contributions underscore how QTERM unifies and extends both classical and quantum learning paradigms within a common theoretical paradigm.

\subsection{Related Work on Quantum Learning Theory}
Similar learning techniques have been explored in the context of quantum models. For instance, structural risk minimization \cite{vapnik1974teoriya} is applied analytically to two quantum linear classifiers in Ref.~\citep{gyurik_structuralrisk_2023}, highlighting the trade-off between training accuracy and generalization performance in parameterized quantum circuits. In parallel, ideas on quantum learning for quantum data have been brought forth in Ref.~\citep{heidari_theoretical_2021}, hereby formalizing a definition of quantum ERM (QERM). Ref.~\cite{padakandla2022pac} generalizes the results of Ref.~\cite{heidari_theoretical_2021} to broader concept classes and a looser version of compatibility. Additionally, Ref.~\citep{heidari2024new} introduces a quantum ERM algorithm that improves sample complexity bounds in the quantum settings through quantum shadows, enabling more efficient empirical loss estimation in quantum classifiers. Moreover, Ref.~\citep{Ciliberto_2020} addresses the question of complexity bounds from a quantum learning perspective. Arunachalam and de Wolf's survey in Ref.~\citep{arunachalam2017guest} analyzes quantum adaptations of classical models such as PAC learning, noting challenges such as sample complexity and measurement incompatibilities due to the fundamental nature of quantum mechanics \cite{salmon2023provableadvantagequantumpac, nayak2024propervsimproperquantum, chung_sample_2021}. Properties of learning from quantum data have been studied from a generalization point of view as well~\cite{caro2022generalization, abbas2021effectivedimensionmachinelearning, GilFuster2024}. Quantum PAC is introduced to address the challenges and adaptations required for the PAC framework when applied to learning from quantum data \cite{Bshouty1999, arunachalam2017guest, Arunachalam2018, padakandla2022pac, heidari_theoretical_2021, heidari2024new}. For instance, in Refs.~\cite{heidari_theoretical_2021, heidari2024new}, the quantum PAC framework accounts for quantum data which is represented as the density operator over a Hilbert space, and secondly, for information loss incurred by the measurement process. Quantum sample complexity for the quantum PAC framework proposed in Ref.~\cite{heidari_theoretical_2021} is known to be higher compared to the classical analogue.

\subsection{Problem Setup and Desiderata}
Consider the task of learning an optimal projector‐valued hypothesis from classical-quantum distribution where quantum states provided in a product‐state form. Without loss of generality, assume $x_1,\dots,x_n$ be i.i.d., classical samples drawn from a continuous distribution with uncountably many possible values. For each sample $x_i$, it is practically impossible to obtain identical copies of the corresponding quantum state. Accordingly, we associated a single copy quantum state $\rho(x_i)$ to each sample $x_i$. The learner’s joint quantum state is therefore the product state
$$\rho(\vec{x}) = \bigotimes_{i=1}^n \rho(x_i),$$
where each $\rho(x_i)$ is available only as a single copy. In particular, there is no way to request $\rho(x_i)\otimes\rho(x_i)$ for any $i$. Consequently, all measurements must be performed on the tensor product of nonidentical, single‐copy states, and must be designed ``gently'' so as not to destroy information needed for subsequent steps. 
Fix a hypothesis class $\mathcal{C} = \{h_c:c\in[m]\}$, where each $h_c$ generates a set of corresponding projectors $\{\Pi(x_i)^{(c)}\}_{i=1}^n$. The quality of a hypothesis is assessed using a loss function defined via projector-valued operations denoted as $\mathrm{Tr}[\rho(x_i) \Pi(x_i)^{(c)}]$. To estimate the loss, the learner performs measurements on $\rho(x_i)$, and these measurements are implemented gently on $\rho(\vec{x})$, so that extracting information about $\tr[\rho(x_i)\,\Pi(x_i)^{(c)}]$ does not destroy the samples that are still needed to test the rest of hypotheses. Within this setting, we introduce TERM parameterized by $\gamma \in \mathbb{R}$, which provides a way to control sensitivity to extreme values, manifested in the form of data irregularities. The tilted risk function is defined as:
\be
\widetilde{R}_{\gamma}(c) = 1 - \frac{1}{\gamma} \log \left( \frac{1}{n} \sum_{i=1}^n e^{\gamma \mathrm{Tr}[\rho(x_i) \Pi(x_i)^{(c)}]} \right).
\ee
For proving our learning guarantee, a certain regime of $\gamma$ is required. For our concentration bounds to hold, we require $|\gamma|<1$, since $\gamma\to 0$ recovers the ordinary mean loss and $\gamma$ too large overemphasizes outliers. In particular, for $\gamma \to 0$, the tilted risk reduces to the standard case of the mean loss $1-\sum_{i=1}^n \mathrm{Tr}[\rho(x_i) \Pi(x_i)^{(c)}]$. The mean loss by default does not overweight extreme values. As $\gamma \to +\infty$, $\widetilde{R}_{\gamma}(c)$ behaves like $1-\max_{i\in [n]} \mathrm{Tr}[\rho(x_i) \Pi(x_i)^{(c)}]$. Large $\gamma$ amplifies variations in the loss landscape, which can be useful for minimizing outliers.

A more familiar way to interpret the role of $\gamma$ is through comparisons with standard norms. For sufficiently small $\gamma$, the tilted empirical risk behaves similarly to an $L_1$-type averaging, providing a balanced estimate of the loss across samples. As $\gamma$ increases, the risk shifts toward an $L_2$-like behavior, amplifying higher-loss samples and emphasizing variance. In the extreme case of very large $\gamma$, the tilted empirical risk resembles a max function (akin to an $L_\infty$ norm), where a few high-loss samples dominate the overall evaluation. See Lemma \ref{lemma:tiltconverg} for a more formal formulation on the convergence of the tilted hyperparameter.

\subsection{Main Contributions}
We address the problem of learning from classical-quantum data by finding an optimal projector-valued hypothesis class. Our main idea is to modify the loss function to be the tilted loss function with a tilted hyperparameter $\gamma$. For the quantum problem, we prove a sample complexity result using the tilted loss for intermediate values of the tilted hyperparameter $\gamma$. In addition, we prove a new PAC generalization bound for the tilted loss minimization for sufficiently small $\gamma$. These two results enable us to prove our third result for a guarantee of agnostic learning for small values of $\gamma$.

We list the main contributions of this paper as follows.
\begin{itemize}
    \item \textbf{Generalization of QERM and formalization of QTERM}. This contribution builds upon the QERM framework introduced in Ref.~\cite{fanizza_learning_2024}, by incorporating the tilted hyperparameter to generalize the theoretical framework of quantum empirical risk minimization. Hereby, we define and formalize QTERM, extending the concept of empirical risk minimization to quantum models trained in quantum states, formally defined in \Cref{def:term}.
    \item \textbf{Learnability of projector-valued functions under sample complexity guarantees}. We derive sample complexity bounds for QTERM, providing theoretical guarantees on the number of quantum state samples required to achieve reliable learning performance. These limits provide a deeper understanding of the interaction between the tilted hyperparameter, the complexity of the model, and the data requirements, expressed in \Cref{thm1:informal}, and more formally in \Cref{theorem:TERM_proj}.
    We state an informal version of the main theorem in \Cref{thm1:informal}. The main theorem answers the following question.
    \begin{itemize}
    \item \textit{How can we efficiently identify the best hypothesis (i.e., a set of projectors) that describes a given quantum system?}
    \item \textit{How many quantum samples (product states) are needed to guarantee a reliable choice?}
    \end{itemize}

    \begin{theorem}[Quantum Tilted Empirical Risk for Projector-Valued Functions, informal] \label{thm1:informal}
    For quantum process learning with product states, our objective is to identify the best hypothesis among $m$ projector-valued functions using tilted empirical risk, a modified risk function controlled by the tilted hyperparameter $\gamma$. With enough quantum samples (approximately $\frac{1}{\varepsilon^2} \log \frac{1}{\delta} \log^2 \frac{1}{\varepsilon}$, plus factors depending on $\gamma$ and $m$), an algorithm can reliably estimate the best hypothesis. Guarantees that the estimated risk is within $\varepsilon$ of the true minimum with a failure probability of at most $\delta$.
    \end{theorem}

    \item \textbf{Generalization bound for TERM.} When the tilted loss class has a controlled covering number and the tilted hyperparameter $\gamma$ is sufficiently small, the tilted empirical risk reliably tracks the true population risk with certain samples; the risk gap stays below $\varepsilon$ with probability at least $1-\delta$. In short, this theorem addresses the following questions.
    \begin{itemize}
        \item \textit{How well does the tilted empirical risk approximate the true risk?}
    \end{itemize}
    
    \begin{theorem}[PAC Bound of TERM, informal] \label{thm2:informal}
    Suppose we minimize TERM over a set of functions. If the function class has a small enough uniform covering number and the tilted hyperparameter $\gamma$ is sufficiently small, then, with high probability over random samples, the tilted empirical risk is close to the true population risk. The probability of a large deviation decays exponentially in the number of samples $n$.
    \end{theorem}

    \item \textbf{Agnostic learning guarantees for quantum hypothesis selection}. We establish agnostic learning bounds for QTERM, extending the theoretical framework to settings where the true hypothesis may not belong to the given class. This result ensures that even in the absence of a perfect hypothesis, we can still approximate the best achievable performance within a controlled error margin. Our guarantees are formalized in \Cref{thm:agnostic_QTERM}, and an informal version is given in \Cref{thm:agnostic_TERM_informal}.

    The follow-up theorem establishes agnostic learning guarantees for the quantum setting, removing the assumption that an optimal hypothesis exists within the given hypothesis class. In particular, this theorem answers the following question.

    \begin{itemize}
        \item \textit{How well can we approximate the best possible hypothesis when no perfect hypothesis exists?}
    \end{itemize}

    \begin{theorem}[Agnostic Learnability of Projector-Valued Function, informal]\label{thm:agnostic_TERM_informal}
    If a class of projector-valued functions $\mathcal{C}$ is not too complex, in the sense that its covering numbers grow slowly with the number of samples, then there exists an algorithm that, given enough training data drawn from an unknown classical-quantum channel, can find a hypothesis $c^* \in \mathcal{C}$ whose performance is close to optimal and whose estimated performance is close to the true value. Moreover, the probability of making a significant error can be made arbitrarily small by choosing the sample size appropriately, depending on the desired accuracy, confidence, and the complexity of $\mathcal{C}$.
    \end{theorem}
\end{itemize}

It is worth highlighting that the theoretical guarantees hold for a certain threshold value of $\gamma$. Throughout our analysis we observe that the learning becomes significantly more difficult for large $\gamma$ values. From a perspective of sample complexity, the bounds we establish in this work hold primarily for sufficiently small $\gamma$. A lesson from our learning guarantee for the generalization error is the exponential dependence on $\gamma$ which could make reliable learning for large $\gamma$ difficult. A contextualization of our contributions is illustrated in \Cref{fig:main}.

\begin{figure*}[t]
    \centering
    \includegraphics[width=1\linewidth]{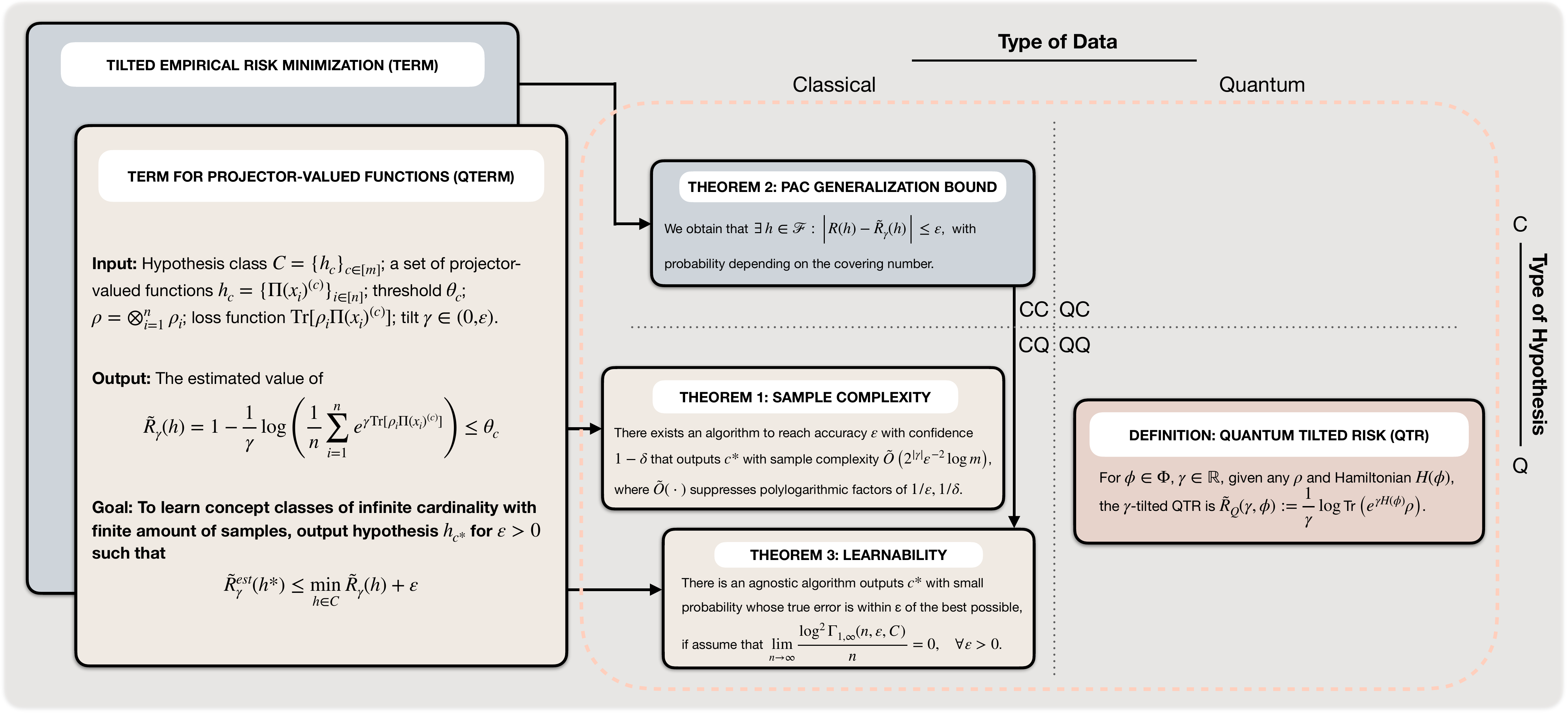}
    \caption{\justifying \textbf{Overview of the learning framework and main theoretical results.} This diagram provides a comprehensive overview of our QTERM framework for learning quantum processes. The left side outlines the key algorithm: which takes as input a hypothesis class, projector-valued functions, product states and corresponding loss function, and a tilting parameter. Through quantum threshold search, the algorithm outputs an optimal hypothesis $h^*$ that minimizes risk. On the right, the figure contextualizes our work by showing how our work extends across different types of data and different types of hypotheses (classical and quantum). Our key theoretical contributions, including provable guarantees on sample complexity, PAC generalization, and agnostic learnability, are summarized by Theorems 1-3, while simultaneously introducing the novel concept of quantum tilt risk (QTR).}
    \label{fig:main}
\end{figure*}

\subsection{Outline of this Work}
This paper is structured as follows. \Cref{sec:background} lays the theoretical foundations of the learning model under consideration, empirical risk minimization, the TERM paradigm, and their quantum extensions to QERM. \Cref{sec:qterm_algorithm} introduces the formal definition of QTERM and derives sample complexity bounds for learning quantum processes. \Cref{sec:pac_bounds} presents new measures of PAC generalization bounds on classical TERM. \Cref{sec:agnostic_learning} introduces learning guarantees under the consideration of agnostic learning. \Cref{sec:qtr} introduces further ideas on defining tilted measures for quantum systems potentially related to Hamiltonian learning. We also comment on the necessity of the different methods of defining the tilt mechanism in different quantum learning settings and applications. This study concludes in~\Cref{sec:conclusion} with a summary of its contributions and an outlook on open directions. Some of the proofs and details on algorithms are deferred to~\Cref{app:proofs} and~\Cref{app:algo}, respectively.

\section{Theoretical Foundations}\label{sec:background}
\subsection{Learning Model Formulation}
There exist multiple model formulations within the learning theoretical frameworks. The well-known PAC learning in the quantum setting is a framework for studying quantum machine learning algorithms, extending the classical PAC learning model to quantum settings. In this framework, the goal is for a learner to approximately learn a target quantum state or quantum process from a finite number of quantum samples, with high probability and within a specified accuracy. Quantum PAC learning has been formalized in multiple settings. 

A supervised quantum machine learning model $\mathcal{M}$ is defined as:
\be
\mathcal{M} = (\mathcal{H}, U(\theta), \mathcal{E}, \mathcal{M}, L),
\ee
where $\mathcal{H}$ is a complex Hilbert space representing the state space, $U(\theta)$ is a parameterized unitary operator acting on $\mathcal{H}$, transforming quantum states. The function $\mathcal{E}: \mathbb{R}^d \to \mathcal{H}$ is an encoding function mapping classical inputs $x$ to quantum states $|\phi\rangle$, $\mathcal{M}: \mathcal{H} \to \mathcal{Y}$ is a measurement operator that extracts classical predictions $\hat{y}$ from the quantum state, $L: \mathcal{Y} \times \mathcal{Y} \to \mathbb{R}$ is a loss function measuring the difference between true labels $y$ and predicted labels $\hat{y}$. The model aims to minimize the expected loss:
\be
\mathcal{L}(\theta) = \frac{1}{N} \sum_{i=1}^N L(y_i, \mathcal{M}(U(\theta) \mathcal{E}(x_i))),
\ee
over a training dataset $\mathcal{D} = \{(x_i, y_i)\}_{i=1}^N$.

Agnostic learning generalizes PAC learning by removing the assumption that the target function is contained within the hypothesis class. Formally, given a hypothesis class $\mathcal{H}$ and a distribution $\mathcal{D}$ over the input space $\mathcal{X}$, the objective in agnostic learning is to find a hypothesis $h \in \mathcal{H}$ that minimizes the expected loss over $\mathcal{D}$, i.e.,  
\be
\inf_{h \in \mathcal{H}} \mathbb{E}_{(x,y) \sim \mathcal{D}} [L(y, h(x))].
\ee
Unlike the PAC setting, where a realizable assumption ensures that there exists a hypothesis with zero or arbitrarily low error, the agnostic framework considers the worst-case scenario where the best hypothesis in $\mathcal{H}$ may still incur significant error. In the quantum setting, this distinction is particularly relevant when learning quantum processes, as the true quantum process may not be exactly representable within the chosen hypothesis class. As shown in Ref.~\cite{fanizza_learning_2024}, when access to training data is constrained by quantum measurement limitations, the agnostic framework provides a formulation for bounding sample complexity, ensuring that learning guarantees remain valid even when the true quantum process deviates from the assumed model class.

\subsection{Tilted Empirical Risk Minimization (TERM)}
Empirical risk minimization is widely used in machine learning where the goal is to optimize model parameters by minimizing the average loss over the training data \citep{vapnik_statistical_1998}. The key idea behind ERM is that, because we do not know the true distribution that generated the data, we use the available training data to estimate the risk by averaging the loss over the dataset. However, ERM has notable limitations, particularly when dealing with outliers or imbalanced data, and generalizing to unseen data. For example, in situations where certain data subgroups are underrepresented or contain outliers, the model may overfit to noisy data or produce unfair solutions, especially if the outliers belong to subgroups that we aim to serve better. To mitigate some of the challenges of ERM, an idea is to use tilted losses as a generalization of this traditional technique. Motivated in large by exponential tilting in deviation theory, works in Refs.~\citep{li2020tilted, li2023tilted} introduce TERM, an extension of ERM with an additional tilted hyperparameter $\gamma$. The flexibility of tilted hyperparameter allows the model to continuously adjust decision boundaries based on the problem settings, offering robustness against outliers, fairness towards underrepresented subgroups, or a balance between both. This approach is especially useful in classification tasks where different groups or data distributions require varying levels of emphasis. Tilted empirical risk bears resemblance with the LogExpSum function, adjusting focus or smoothing behavior through log-exp functions.

The setting of supervised learning can be defined as follows. Let $D = \{(x_1, y_1), (x_2, y_2), \ldots, (x_n, y_n)\}$ denote a training set, where $x_i \in \mathcal{X}$ represents the input feature vector and $y_i \in \mathcal{Y}$ denotes the corresponding output label for each instance $i \in \{1, 2, \ldots, n\}$. The objective of supervised learning is to learn a mapping function $f: \mathcal{X} \rightarrow \mathcal{Y}$ that minimizes the expected prediction error on the output space. This is typically achieved by optimizing a loss function $L(y_i, \hat{y}_i)$, where $\hat{y}_i = f(x_i)$ is the predicted label for the input $x_i$. The optimization process can be expressed as:
\be
\hat{f} = \arg\min_{f \in \mathcal{F}} \frac{1}{n} \sum_{i=1}^{n} L(y_i, f(x_i)),
\ee
where $\mathcal{F}$ denotes the hypothesis space of possible functions. The trained model $\hat{f}$ is then expected to generalize to unseen data, allowing for accurate predictions on new instances $x\in\mathcal{X}$. Generalization error quantifies how well the learned model $\hat{f}$ performs on new samples drawn from the same data distribution $\mathcal{D}$ as the training data. Generalization error $\mathbb{E}$ can be expressed as: 
\be
\mathbb{E}_{(x, y) \sim \mathcal{D}} \left[ L(y, \hat{y}) \right] - \frac{1}{n} \sum_{i=1}^{n} L(y_i, \hat{y}_i),
\ee
where $L(y, \hat{y})$ is the loss function measuring the difference between the true label $y$ and the predicted label $\hat{y} = \hat{f}(x)$. Minimizing generalization error is the primary goal of a supervised learner. Following the above, we can now define ERM more formally.
\begin{definition}[Empirical Risk Minimization --- ERM \citep{vapnik_statistical_1998}]
For a hypothesis $h({\bf x}^{(i)})$, empirical risk minimization, the average loss over the training data, is defined as
\be
\Bar{R}(\theta):= \frac{1}{N} \sum_{i \in [N]} \mathcal{L}(h({\bf x}^{(i)}), y^{(i)}, \theta), 
\ee
where $\mathcal{L}(h(x_i), y_i, \theta)$ is the loss function that quantifies the distance between the prediction $h(x_i)$ and true label $y_i$ and parameter $\theta$, $N$ is the number of training data points.
\end{definition}
The performance of ERM is influenced by \textit{sample complexity}, which refers to the number of samples required to achieve a certain level of generalization error. In order to establish good generalization performance, it is important to understand the relationship between the complexity of the hypothesis space and the amount of training data. While ERM focuses on minimizing the average loss over the training dataset, TERM introduces a modified approach by applying the exponential tilting technique to ERM, assigning different level of emphasis to the loss of samples. TERM can be defined as below.
\begin{definition}[Tilted Empirical Risk Minimization --- TERM \citep{li2020tilted,li2023tilted}]\label{def:term}
For $\gamma \in \mathbbm{R}^{ \backslash 0}$, the $\gamma$-tilted loss in ERM is defined as the tilted empirical risk minimization, given by
\be
\widetilde{R}_{\gamma}(h):= \frac{1}{\gamma}\log \left(\frac{1}{N} \sum_{i \in [N]} e^{\gamma \mathcal{L}(h({\bf x}^{(i)}), y_i, \theta)} \right),
\ee    
where $\mathcal{L}(h({\bf x}^{(i)}), y_i, \theta)$ is the loss function on hypothesis $h({\bf x}^{(i)})$, true label $y_i$ and parameter $\theta$, and N is the number of training samples.
\end{definition}

The TERM framework introduces a flexible tilting tool to address the shortcomings of ERM, offering more robust solutions by adjusting the sensitivity of model to outliers using the tilted hyperparameter. Solutions such as TERM are crucial to complementing ERM, because of its sensitivity to overfitting. TERM, in turn, ensures lower generalization error. The generalization error of QTERM for both bounded and unbounded loss functions has been studied with respect to both information-theoretic as well as uniform bounds \cite{aminian2024generalizationerrortiltedempirical}. Across both bounded and unbounded loss functions, the convergence rates for the upper bounds consistently exhibit a rate of $O(1/\sqrt{n})$, where $n$ is the number of training samples.

\subsection{Mathematical Framework and Assumptions}
In classical learning, the effectiveness of TERM relies on specific conditions related to the complexity of the hypothesis class. A key tool we use to analyze this complexity is the \textit{covering number}, a standard measure of the effective size of a concept class formalized by Ref. \cite{anthony2009neural}. 

\begin{definition}[Covering Number]
Let $\mathcal{G} \subseteq L(\mathcal{H})^{\mathcal{X}}$ be a class of functions mapping from a space $\mathcal{X}$ to bounded linear operators on a Hilbert space $\mathcal{H}$, and let $\varepsilon \in (0, 1]$. The \emph{covering number} of $\mathcal{G}$ is defined as:
\be
\Gamma_{1,q}(n, \varepsilon, \mathcal{G}) := \max \left\{ N_{\text{in}}(\varepsilon, \mathcal{G}, \| \cdot \|_{1,q,\bar{x}}) \mid \bar{x} \in \mathcal{X}^n \right\},
\ee
where $N_{\text{in}}(\varepsilon, \mathcal{G}, d)$ is the minimal cardinality of any \emph{internal $\varepsilon$-cover} of $\mathcal{G}$ under the pseudometric $d$, given by the seminorm $\| \cdot \|_{1,q,\bar{x}}$ and $\bar{x} = (x_1, \dots, x_n) \in \mathcal{X}^n$ a sample of size $n$.
\end{definition}
To treat concept classes whose hypotheses take values in an operator algebra rather than in $\mathbb{R}$, we introduce an \textit{operator-class covering number}, defined with a seminorm specifically tailored to operator-valued functions:
\begin{definition}[Operator-class covering number]
Let $\mathcal{C}\subseteq \mathcal{L}(\mathcal{H})^\mathcal{X} $ be a hypothesis class, and let $\|\cdot\|_{p,q,\vec{x}} $ denote a seminorm over datasets  $\vec{x}\in \mathcal{X}^n $. Given $\varepsilon>0$, the operator-class covering number $\Gamma_{p,q}(n,\varepsilon,\mathcal{C})$ quantifies the complexity of $\mathcal{C}$, defined as the minimal cardinality of an internal $\varepsilon $-cover.
\end{definition}

\begin{definition}[Exponentially tilted loss-function class]
Let $\mathcal{F}\subseteq Y^{\mathcal{X}}$ be any function class and let $L : Y\times Y \to [0,\,a]$ be a bounded loss.
For any tilting parameter $\gamma>0$ we define the ``exponential loss-function class'' associated with $(\mathcal{F},L,\gamma)$ by

$$
\mathcal{G}_{\mathcal{F},L,\gamma} = \{g_h : \mathcal{X}\times Y \to [1,\,e^{\gamma a}]\,|\, h\in\mathcal{F}, \,g_h(x,y)=\exp (\gamma\,L(y, h(x)))\}.
$$
\end{definition}
This mirrors the structure of the \textit{induced loss-function class} in Definition 6 in \cite{fanizza_learning_2024} while replacing the original loss with the exponential tilting loss. The authors establish a quantum variant of ERM for projector-value concept class as below, which is the basis of our result for extension to quantum variant of \textit{tilted} ERM.
\begin{fact}[Quantum empirical risk minimization (QERM) for projector-valued functions, Theorem 1 in \cite{fanizza_learning_2024}]\label{fact:QERM}
Let us consider a product state $ \rho = \rho_1 \otimes \cdots \otimes \rho_n$ and a collection of projectors $ \{\Pi_1^{(c)}, \dots, \Pi_n^{(c)}\}_{c=1}^m $, where
$$
\mu_c = \frac{1}{n} \sum_{i=1}^n \mathrm{Tr}[\rho_i \Pi_i^{(c)}].\footnote{For simplicity, we will subsequently denote $\rho(x_i)$ as $\rho_i$, $\Pi(x_i)$ as $\Pi_i$.}
$$
There exists an algorithm that outputs both an index $ c^* $ and an estimated value $ \hat{\mu}_{c^*} $ such that
$$
\Pr\left(\left|\hat{\mu}_{c^*} - \max_{c \in [m]} \mu_c \right| \geq \varepsilon \ \cup \ \left|\hat{\mu}_{c^*} - \mu_{c^*} \right| \geq \varepsilon\right) \leq \delta,
$$
if the number of samples $n$ is sufficiently large. Specifically, $n$ can be chosen as
$$
n = O\left(\frac{\log \frac{1}{\delta} \max\left(\log \frac{m}{\delta}, \log^2(em)\right)}{\varepsilon^2}\right).
$$
\end{fact}

The next lemma summarise a key property of the TERM framework.

\begin{lemma}[Convergence of $\widetilde{R}_\gamma(\theta)$, Lemma 4 in \cite{li2020tilted,li2023tilted}]\label{lemma:tiltconverg}
    Consider that the loss function $L(x_i;\theta)$ belongs to the differentiability class $C^1$ (i.e., continuously differentiable) with respect to $\theta \in \Theta \subseteq \mathbb{R}^d$ for $i \in [N]$. Tilted empirical risk $\widetilde{R}_\gamma(\theta)$ converges to min-loss, ERM, max-loss as t moves from $-\infty$ to $\infty$.
    
\end{lemma}
Hence, noted that the parameter $\gamma$ controls the focus of the risk function as tilting hyperparameter varies.

\subsection{Quantum Threshold Search}
Quantum threshold search proposed in Ref.\cite{fanizza_learning_2024} operates on a product state of non-identical quantum states, a crucial generalization from prior work \cite{B_descu_2024} that required identical copies. This algorithm enables learning from classical-quantum data in practical scenarios where input cannot be controlled. It's mainly designed to solve a specific learning problem: given an unknown quantum state, the goal is to find an `observable-threshold' pair from a predefined set where the expectation value of observable exceeds that threshold on the state. Its properties are defined in Lemma \ref{lemma:Qthreshold_search}.

\begin{lemma}[Quantum threshold search on nonidentical states, Lemma 1 in \cite{fanizza_learning_2024}]\label{lemma:Qthreshold_search} 
Suppose that given the following inputs:
\begin{enumerate}[label=(\alph*)]
\item A single product state $\rho = \rho_{1} \otimes \cdots \otimes \rho_{n} \in \mathcal{D}(\mathcal{H}^{(d)})^{\otimes n}$, with each $\rho_i$ a (possibly nonidentical) qudit state.
\item Parameters $\varepsilon \in [0,1]$.
\item A collection of lists of projectors $\{\Pi_{1}^{(c)}, \dots, \Pi_{n}^{(c)}\}_{c\in[m]}$.
\item A set of known thresholds $\{\theta_c\}_{c\in[m]}$ where $\theta_c \in [0,1]$.
\end{enumerate}
Assume there is at least one $c$ for which
$$
\frac{1}{n} \sum_{i=1}^n \mathrm{Tr}\left[\Pi_i^{(c)} \rho_i\right]>\theta_c.
$$
Then there exists an algorithm such that
\begin{enumerate}[label=(\alph*)]
\item Performs a two-outcome measurement $\left\{B_c, \mathbf{1}-B_c\right\}$ at each step $c=1, \ldots, m$, based on the projectors $\left\{\Pi_1^{(c)}, \ldots, \Pi_n^{(c)}\right\}$.
\item  If the measurement accepts $B_c$, then the algorithm halts and outputs $c$, otherwise it proceeds to $c+1$.
\end{enumerate}
If the condition holds: $ \left(\log m+C_2\right)^2<C_1 n \varepsilon^2 $ for appropriate constants $C_1, C_2>0$, it outputs $c$ such that
$$
\frac{1}{n} \sum_{i=1}^n \mathrm{Tr}\left[\Pi_i^{(c)} \rho_i\right] \geq \theta_c-\varepsilon
$$
with probability at least $0.03$.
\end{lemma}

The quantum threshold search algorithm is formulated as \cref{alg:threshold_search} in \cref{app:algo}. This algorithm runs a sequence of two-outcome POVMs $B_c$ on the same state, halting as soon as it finds a projector list whose sample average exceeds its threshold, or else certifying that none do.

This technique leverages a generalization of standard measurement, the gentler technique from Refs. \cite{B_descu_2024,fanizza_learning_2024} to achieve improved performance. Unlike standard projective measurements, which can cause significant state collapse, gentle measurements allow us to check if an expectation value is above or below a threshold without fundamentally altering the quantum state. This is crucial for our setting, which involves products of non-identical states. We then optimize the threshold for each concept using binary search. Specifically, we begin with a candidate threshold of 1/2. For each iteration, we run \cref{alg:threshold_search} on a new batch of data to check if a candidate below the current threshold exists. Based on this result, we select either the lower or upper half of the candidate interval, updating the threshold to 1/4 or 3/4, respectively. This process continues until we approximate the empirical risk to a desired precision.

The next section makes use of the lemmas and propositions defined above, in order to define and formalize TERM in the presence of quantum states.

\section{Tilted Empirical Risk for Quantum Learning (QTERM)}\label{sec:qterm_algorithm}

In this section, we introduce a TERM version tailored for quantum learning problem which we refer to as QTERM from quantum TERM. Specifically, this framework focuses on learning from quantum data by using collections of projector-valued hypotheses. We define a tunable loss function based on the probability of not accepting these projector-valued functions with a tilting parameter.

\begin{definition}[Tilted Empirical Risk Minimization for Projector-Valued Functions]
Let $m$ hypotheses be represented by projector-valued functions $h_{c,s}=\{\Pi_{s,1}^{(c)}, \dots, \Pi_{s,l}^{(c)}\}$, where $s$ indexes blocks of size $l$. Given a product state $\rho = \rho_1 \otimes \cdots \otimes \rho_n$, with $\rho_s \in \mathcal{D}(\mathcal{H}^{(d)})$, the \emph{tilted empirical risk} of hypothesis $c$ for tilting parameter $\gamma \in \mathbb{R}$ is:
\be
\widetilde{R}_{\gamma}(h) := 1 - \frac{1}{\gamma} \log \left( \frac{1}{n} \sum_{i=1}^n e^{\gamma \mathrm{Tr}[\rho_i \Pi_i^{(c)}]} \right).
\ee
\end{definition}

This definition provides the theoretical foundation for our QTERM framework, enabling us to identify an optimal hypothesis that approximately minimizes the true tilted risk $R(h)$ with tunable losses.

We will adopt this method to implement TERM within our framework. Its properties are defined in the following proposition. It is primarily adapted from Lemma 2 and Lemma 5 in \cite{fanizza_learning_2024}, with necessary adjustments made—mainly replacing the empirical risk with the tilted empirical risk.

\begin{proposition}[Conditions on QTERM given sufficiently large product states]\label{prop:conditions} 
Given product states divided into $n/l=2Tk$ blocks as $\rho = \rho_{1} \otimes \cdots \otimes \rho_{2Tk}, \, \text{where} \, \rho_s = \rho_{s,1} \otimes \cdots \otimes \rho_{s,l},$ and a collection of lists of projectors $\{\Pi_{s,1}^{(c)}, \dots, \Pi_{s,l}^{(c)}\}_{c\in[m],\,s\in[2Tk]}$, assume that the following conditions hold:
\begin{enumerate}[label=(\roman*)]
    \item For proper constants $C_1,\,C_2>0$, let
    \begin{align}
    l > \frac{(\log m +C_2)^2}{C_1 \varepsilon^2}.
    \end{align}
    \item Simultaneously, 
    \begin{align}
    l > \frac{\log (Tk/\delta)}{\varepsilon^2}
    \end{align}
    for large enough $T=O(\log\frac{1}{\varepsilon})$ and $k=O(\log\frac{1}{\delta}\log\frac{1}{\varepsilon})$.
    \item The tilted empirical risk of each hypothesis $c$ for $\gamma$ is defined as
    \be
        \widetilde{R}_{\gamma}(c) = 1-\mu_c(\gamma) := 1- \frac{1}{\gamma} \log \left
        (\frac{1}{l}\sum_{j=1}^l e^{\gamma \mathrm{Tr}[\rho_{s,j} \Pi_{s,j}^{(c)}]}\right),
    \ee
    where the approximations of $\frac{1}{\gamma} \log (\frac{1}{l}\sum_{j=1}^l e^{\gamma E_{\rho_{s,j}}[\Pi_{s,j}^{(c)}]})$ with respect to the blocks of states for all $c\in[m]$, are $\mu_c(\gamma)\in[0,1]$, satisfying that
    \begin{align}
    \left| \frac{1}{\gamma} \log \left( \frac{1}{l}\sum_{j=1}^l e^{\gamma E_{\rho_{s,j}}[\Pi_{s,j}^{(c)}]}\right) - \mu_c(\gamma) \right| \leq \varepsilon/4.
    \end{align}
\end{enumerate}
There exists an algorithm that outputs both an index $ c^* $ and an estimated value $ \hat{\mu}_{c^*}(\gamma) $ such that
$$
\Pr\left(\left|\hat{\mu}_{c^*}(\gamma) - \max_{c \in [m]} \mu_c(\gamma) \right| \geq \varepsilon \ \cup \ \left|\hat{\mu}_{c^*}(\gamma) - \mu_{c^*}(\gamma) \right| \geq \varepsilon\right) \leq \delta
$$
for large enough $n =\Omega(1/\varepsilon^2)$.
\end{proposition}

To find the tilted empirical risk minimizer in the quantum learning case, we prove the following theorem. By systematically minimizing the TERM, the algorithm identifies the optimal hypothesis (indexed as $c^*$) and provides a sample complexity measure guaranteeing high-confidence probability when the sample size $n$ is sufficiently large. Combined with Lemma~\cref{lemma:Qthreshold_search}, \cref{alg:threshold_search} retains a constant success probability in finding a projector meeting its threshold criterion under proper conditions in $m$ and some constants.

\begin{theorem}[Tilted empirical risk minimization for projector-valued functions]\label{theorem:TERM_proj}
Let there be $m$ possible concepts (or hypotheses), each of which is defined by a set of projector-valued functions $h_c = \{\Pi_{1}^{(c)}, \dots, \Pi_{n}^{(c)}\}$ and given a product state $ \rho = \rho_1 \otimes \cdots \otimes \rho_n $ where each $\rho_i\in \mathcal{D}(\mathcal{H}^{(d)})$, with 
$$
\mu_c(\gamma) = \frac{1}{\gamma} \log \left( \frac{1}{n} \sum_{i=1}^n e^{\gamma \mathrm{Tr}[\rho_i \Pi_i^{(c)}]} \right),
$$
where $|\gamma|\in (0,\varepsilon)$. There exists an algorithm that outputs both an index $ c^* $ and an estimated value $ \hat{\mu}_{c^*}(\gamma) $ such that
$$
\Pr\left(\left|\hat{\mu}_{c^*}(\gamma) - \max_{c \in [m]} \mu_c(\gamma) \right| \geq \varepsilon \ \cup \ \left|\hat{\mu}_{c^*}(\gamma) - \mu_{c^*}(\gamma) \right| \geq \varepsilon \right) \leq \delta
$$
for sufficiently large $n$ with the probability of error less than $\delta$. One can choose large enough
$$
n = \frac{1}{\varepsilon^2} \log \frac{1}{\delta} \log^2 \frac{1}{\varepsilon} \times O\left(\max \left( \frac{(e^{|\gamma|}-1)^2}{\gamma^2} \log \left(\frac{m}{\delta} \log \frac{1}{\delta} \log^2 \frac{1}{\varepsilon}\right),\, (\log m+C_1)^2 \right)\right),
$$
for appropriate $C_1>0$.
\end{theorem}

\begin{proof}
Let us begin by choosing $n$ large enough so that $n \geq 6Tkl$. With the conditions in Proposition \ref{prop:conditions}, let $l > \frac{\log (Tk/\delta)}{\varepsilon^2}$ for large enough $T=O(\log\frac{1}{\varepsilon})$ and $k=O(\log\frac{1}{\delta}\log\frac{1}{\varepsilon})$. Each of the $n$ observations corresponds to a single quantum state $\rho_i$. Since $K = 2Tk$ in Proposition~\ref{Prop:Hoeffdings2}, sample without replacement $2Tk$ blocks of length $l$ from $[n]$. Let the $s$th block form the product state 
$$
\rho = \rho_{s,1} \otimes \cdots \otimes \rho_{s,l}, \quad s=1,\cdots,2Tk
$$
and for each concept $c\in [m]$, the set of projectors
\begin{align}
\{\Pi_{s,1}^{(c)}, \dots, \Pi_{s,l}^{(c)}\},\quad s=1,\cdots,2Tk. 
\end{align}
Within block $s$, we identify
\begin{align*}
X^{(c)}_{s,i}(\gamma) := e^{\gamma \mathrm{Tr}[\rho_{s,i} \Pi_{s,i}^{(c)}]}.
\end{align*}
Step 1: By substituting $X^{(c)}_{s,j}(\gamma)$, we can apply the Proposition \ref{Prop:Hoeffdings2} to provide the bound on the estimation such that 
\begin{align}\label{Ineq:error}
\left\lvert \frac{1}{l}\sum_{j=1}^l X^{(c)}_{s,j}(\gamma) - \frac{1}{n} \sum_{i=1}^n e^{\gamma \mathrm{Tr}[\rho_i \Pi_i^{(c)}]} \right\rvert \leq |e^\gamma -1|\sqrt{\frac{2 \log (4Tkm/\delta_1)}{l}}, \quad \forall c\in[m], \, s\in [2Tk],
\end{align}
with the probability of error $p_{\text{err}}^{(1)}$ less than $\delta_1$. Regarding the bound on the tilted risk estimator, for $\gamma<0$, we have
\begin{align}\label{Bound:neg}
& \left| \frac{1}{\gamma} \log \frac{1}{l}\sum_{j=1}^l X^{(c)}_{s,j}(\gamma) - \frac{1}{\gamma}\log \frac{1}{n} \sum_{i=1}^n e^{\gamma \mathrm{Tr}[\rho_i \Pi_i^{(c)}]} \right| \\\nonumber
(\text{Log-Lipschitz}) & \leq \frac{\exp(-\gamma)}{|\gamma|} \left| \frac{1}{l}\sum_{j=1}^l X^{(c)}_{s,j}(\gamma) - \frac{1}{n} \sum_{i=1}^n e^{\gamma \mathrm{Tr}[\rho_i \Pi_i^{(c)}]} \right|\\\nonumber
(\text{By Prop. \ref{Prop:Hoeffdings2}}) &\leq \frac{\exp(-\gamma) (1-\exp(\gamma))}{|\gamma|} \sqrt{\frac{2 \log (4Tkm/\delta_1)}{l}}.
\end{align}
Similarly, for $\gamma>0$, 
\begin{align}\label{Bound:pos}
\left| \frac{1}{\gamma} \log \frac{1}{l}\sum_{j=1}^l X^{(c)}_{s,j}(\gamma) - \frac{1}{\gamma}\log \frac{1}{n} \sum_{i=1}^n e^{\gamma \mathrm{Tr}[\rho_i \Pi_i^{(c)}]} \right| \leq \frac{(\exp(\gamma)-1)}{|\gamma|} \sqrt{\frac{2 \log (4Tkm/\delta_1)}{l}}.
\end{align}
Combining the bound in (\cref{Bound:neg}) and (\cref{Bound:pos}), we can obtain the error bound $\varepsilon/4 = \frac{(\exp(|\gamma|)-1)}{|\gamma|} \sqrt{\frac{2 \log (4Tkm/\delta_1)}{l}}$ with probability $p_{\text{err}}^{(1)} \leq 4Tkm e^{-l \gamma^2 \varepsilon^2/32(e^{|\gamma|} -1)^2}$, and it follows that condition (iii) in Proposition~\ref{prop:conditions} is satisfied.
\\\\
Step 2: We apply \textbf{ThresholdSearch} (Lemma~\ref{lemma:Qthreshold_search}) to the states $\rho_s$ for each odd $s$. For each concept $c\in [m]$, we input the following two lists of projectors into the search, each with precision parameter $\varepsilon / 4$ and threshold $\lambda_c$:

(1) $\left\{\Pi_{s, 1}^{(c)}, \Pi_{s, 2}^{(c)}, \ldots, \Pi_{s, l}^{(c)}\right\}$ with threshold $\lambda_c+7 / 4 \varepsilon$,

(2) $\left\{\mathbf{1}-\Pi_{s, 1}^{(c)}, \mathbf{1}-\right.$ $\left.\Pi_{s, 2}^{(c)}, \ldots, \mathbf{1}-\Pi_{s, l}^{(c)}\right\}$ with threshold $1-\lambda_c-7 / 4 \varepsilon$,

If \textbf{ThresholdSearch} returns a concept $c$ for a particular odd $s$, we perform the measurements $\{\Pi_{s+1,j}^{(c)}\}_{j\in[l]}$

on $\{\rho_{s+1,j}\}_{j\in [l]}$, with a list of resulting random variables $\left\{\mathrm{Tr}[\rho_{s+1,j} \Pi_{s+1,j}^{(c)}]\right\}_{j\in[l]}$ and define $X_{s+1,j}^{(c)}(\gamma):= e^{\gamma \mathrm{Tr}[\rho_{s+1,j} \Pi_{s+1,j}^{(c)}]}$ for all $j\in [l]$. If $\left| \frac{1}{\gamma} \log \frac{1}{l} \sum_{j=1}^l \mathrm{X}^{(c)}_{s+1,j} (\gamma) -\lambda_c\right|>\varepsilon$ is satisfied, the check is considered passed, and we output $c$. If none of the checks pass, we declare a failure.

Next we analyze two error scenarios and the corresponding probabilities:

(a) Suppose there exists $c$ such that $\left|\mu_c-\lambda_c\right| \geq 2 \varepsilon$. By condition (iii) in Lemma 2 in \cite{fanizza_learning_2024}, it follows that \\$\left|1 / l \sum_{j=1}^l E_{\rho_{s, j}}\left[\Pi_{s, j}^{(c)}\right]-\lambda_c\right| \geq 7 / 4 \varepsilon$. This implies:
$$
1 /l \sum_{j=1}^l E_{\rho_{s, j}}\left[\Pi_{s, j}^{(c)}\right] \geq \lambda_c+7 / 4 \varepsilon, \quad \text{or} \quad 1 / l \sum_{j=1}^l E_{\rho_{s, j}}\left[1- \Pi_{s, j}^{(c)}\right] \leq 1-\lambda_c-7 / 4 \varepsilon.
$$
Since \textbf{ThresholdSearch} guarantees that if the projector or its complement exceeds the specified threshold, the algorithm returns that concept with probability at least $0.03$. We conclude that each time \textbf{ThresholdSearch} is performed on the pair of projector lists, it outputs a concept $c$ such that $\left|1 / l \sum_{j=1}^l E_{\rho_{s, j}}\left[\Pi_{s, j}^{(c)}\right]-\lambda_c\right|>6 / 4 \varepsilon$. For such a concept $c$, and for any $\gamma \in \mathbb{R}$, consider the quantity
\be
\Delta(\gamma) = \left| \frac{1}{\gamma} \log \frac{1}{l}\sum\limits_{j=1}^l \exp\left(\gamma E_{\rho_{s, j}}[\Pi_{s, j}^{(c)}]\right) - \frac{1}{l} \sum\limits_{j=1}^l E_{\rho_{s, j}}[\Pi_{s, j}^{(c)}] \right|
\ee
Noted that $\bar{E}_{s,j} = \sum\limits_{j=1}^l E_{\rho_{s, j}}[\Pi_{s, j}^{(c)}]$. To bound $\Delta(\gamma)$ from above, let
\be
\delta_j = E_{\rho_{s, j}}[\Pi_{s, j}^{(c)}] - \bar{E}_{s,j},
\ee
such that $\frac{1}{l}\sum_{j=1}^l \delta_j = 0$. 
Therefore, we obtain
\be
\frac{1}{\gamma}\,\log\left(\frac{1}{l}\,\sum_{j=1}^l e^{\gamma E_{\rho_{s, j}}[\Pi_{s, j}^{(c)}]}\right)
=
\frac{1}{\gamma}\,\left[\gamma\bar{E}_{s,j}
+
\log\left(\frac{1}{l}\,\sum_{j=1}^l e^{\gamma\delta_j}\right)\right]
= \bar{E}_{s,j} +
\frac{1}{\gamma}\,\log\left(\frac{1}{l}\,\sum_{j=1}^l e^{\gamma\delta_j}\right)
= \bar{E}_{s,j} + \Delta(\gamma).
\ee
Since $E_{\rho_{s, j}}[\Pi_{s, j}^{(c)}] \in [0,1]$ for each $j$, then $\bar{E}_{s,j}\in [0,1]$ and thus each deviation $\delta_j\in[-\bar{E}_{s,j} , 1 - \bar{E}_{s,j}]\subseteq [-1, 1]$. Consequently, $\{\delta_j\}$ is a zero-mean set bounded by $\pm 1$.

Applying Hoeffding’s Lemma, for every $\gamma\in\mathbb{R}$, 
\be
\frac{1}{l}\,\sum_{j=1}^l e^{\gamma\delta_j}
\le
\exp\left(\frac{\gamma^2}{8}\right).
\ee
Taking logarithms and dividing by $\gamma>0$
\be
0 \leq \Delta(\gamma) \leq \frac{1}{\gamma}\log\left(\frac{1}{l}\,\sum_{j=1}^l e^{\gamma\delta_j}\right)
\le
\frac{\gamma}{8}
\ee
By an analogous argument (applied to $-\gamma$) one obtains $\Delta(\gamma)\ge -\frac{|\gamma|}{8}$. Thus for every real $\gamma\neq 0$, $\left|\Delta(\gamma)\right| \le \frac{|\gamma|}{8}$.
    In particular, if we desire $\left|\Delta(\gamma)\right| \in O(\varepsilon)$, it suffices to choose $|\gamma| \in (0,\varepsilon)$, i.e., $\gamma$ is sufficiently small.

Finally, we apply multiplicative Chernoff bound to have
\begin{align}\label{Ineq:pFPC}
& \Pr \left(\left|\frac{1}{\gamma} \log \frac{1}{l} \sum_{j=1}^l X_{s+1,j}^{(c)}(\gamma) - \lambda_c \right| \leq \varepsilon\right) \\\nonumber
& = \Pr \left(\left|\frac{1}{\gamma} \log \frac{1}{l} \sum_{j=1}^l X_{s+1,j}^{(c)}(\gamma) - \frac{1}{l}\sum_{j=1}^l E_{\rho_{s,j} }[\Pi_{s,j}^{(c)}] + \frac{1}{l}\sum_{j=1}^l E_{\rho_{s,j} }[\Pi_{s,j}^{(c)}] - \lambda_c \right| \leq \varepsilon\right)\\\nonumber
(\text{Triangle Ineq.})& \leq \Pr \left(\left||\frac{1}{\gamma} \log \frac{1}{l} \sum_{j=1}^l X_{s+1,j}^{(c)}(\gamma) - \frac{1}{l}\sum_{j=1}^l E_{\rho_{s,j} }[\Pi_{s,j}^{(c)}] | - | \frac{1}{l}\sum_{j=1}^l E_{\rho_{s,j} }[\Pi_{s,j}^{(c)}] - \lambda_c |\right| \leq \varepsilon\right)\\\nonumber
& \leq \Pr \left(\left|\frac{1}{\gamma} \log \frac{1}{l} \sum_{j=1}^l X_{s+1,j}^{(c)}(\gamma) - \frac{1}{l}\sum_{j=1}^l E_{\rho_{s,j} }[\Pi_{s,j}^{(c)}] \right|\geq \varepsilon/2\right)\\\nonumber
(\text{Let $E_{s,j}:= E_{\rho_{s,j}}[\Pi_{s,j}^{(c)}]$}) & = \Pr \left(\left| \frac{1}{\gamma} \log \frac{1}{l} \sum\limits_{j=1}^l X_{s+1,j}^{(c)}(\gamma) -  \frac{1}{\gamma} \log \frac{1}{l}\sum\limits_{j=1}^l e^{\gamma E_{s,j}} + \frac{1}{\gamma} \log \frac{1}{l}\sum\limits_{j=1}^l e^{\gamma E_{s,j}} - \frac{1}{l}\sum_{j=1}^l E_{\rho_{s,j} } \right|\geq \varepsilon/2 \right)\\\nonumber
(\text{Triangle Ineq.})& \leq \Pr \left( \left| \frac{1}{\gamma} \log \frac{1}{l} \sum\limits_{j=1}^l X_{s+1,j}^{(c)}(\gamma) -  \frac{1}{\gamma} \log \frac{1}{l}\sum\limits_{j=1}^l e^{\gamma E_{s,j}} \right| + \left| \frac{1}{\gamma} \log \frac{1}{l}\sum\limits_{j=1}^l e^{\gamma E_{s,j}} - \frac{1}{l}\sum_{j=1}^l E_{\rho_{s,j} } \right| \geq \varepsilon/2 \right)\\\nonumber
(\text{Jensen's Ineq.}) & \leq \Pr \left(\left| \frac{1}{\gamma} \log \frac{1}{l} \sum\limits_{j=1}^l X_{s+1,j}^{(c)}(\gamma) -  \frac{1}{\gamma} \log \frac{1}{l}\sum\limits_{j=1}^l e^{\gamma E_{s,j}} \right| \geq \varepsilon/4 \right) \\\nonumber
(\text{Log-Lipschitz})& \begin{cases}
    \leq \Pr\left( \left| \frac{1}{l}\sum_{j=1}^l X^{(c)}_{s+1,j}(\gamma) - \frac{1}{l} \sum_{i=1}^l e^{\gamma E_{s,j}} \right| \geq |\gamma| e^{\gamma}\varepsilon/4 \right) \quad (\gamma<0 )\\\nonumber
    \leq \Pr\left( \left| \frac{1}{l}\sum_{j=1}^l X^{(c)}_{s+1,j}(\gamma) - \frac{1}{l} \sum_{i=1}^l e^{\gamma E_{s,j}} \right| \geq |\gamma|\varepsilon/4 \right) \quad (\gamma>0 )
\end{cases}\\\nonumber
(\text{By Prop. \ref{lemma:NEestimationExp}}) & \begin{cases}
\leq 2 e^{-l |\gamma| e^{\gamma} \varepsilon^2/12(1-e^{\gamma})} \quad (\gamma<0)\\
\leq 2 e^{-l |\gamma| \varepsilon^2/12(e^{\gamma}-1)} \quad (\gamma>0)
\end{cases}\\\nonumber
& = 2 e^{-l |\gamma| \varepsilon^2/12(e^{|\gamma|}-1)}
\end{align}

Thus, it shows that when $\left|\mu_c-\lambda_c\right| \geq 2 \varepsilon$, the probability of the check passing is at least $1- 2 e^{-l |\gamma| \varepsilon^2/12(e^{|\gamma|}-1)}$.

(b) When a concept is chosen with $\left|\mu_c-\lambda_c\right| \leq \varepsilon/2$, the proof is similar as (\cref{Ineq:pFPC}) such that
\begin{align}
& \Pr \left(\left|\frac{1}{\gamma} \log \frac{1}{l} \sum_{j=1}^l X_{s+1,j}^{(c)}(\gamma) - \lambda_c \right| \geq \varepsilon\right) \\\nonumber
& = \Pr \left(\left|\frac{1}{\gamma} \log \frac{1}{l} \sum_{j=1}^l X_{s+1,j}^{(c)}(\gamma) - \frac{1}{l}\sum_{j=1}^l E_{\rho_{s,j} }[\Pi_{s,j}^{(c)}] + \frac{1}{l}\sum_{j=1}^l E_{\rho_{s,j} }[\Pi_{s,j}^{(c)}] - \lambda_c \right| \geq \varepsilon\right)\\\nonumber
(\text{Triangle Ineq.})& \leq \Pr \left(|\frac{1}{\gamma} \log \frac{1}{l} \sum_{j=1}^l X_{s+1,j}^{(c)}(\gamma) - \frac{1}{l}\sum_{j=1}^l E_{\rho_{s,j} }[\Pi_{s,j}^{(c)}] | + | \frac{1}{l}\sum_{j=1}^l E_{\rho_{s,j} }[\Pi_{s,j}^{(c)}] - \lambda_c | \geq \varepsilon\right)\\\nonumber
& \leq \Pr \left(|\frac{1}{\gamma} \log \frac{1}{l} \sum_{j=1}^l X_{s+1,j}^{(c)}(\gamma) - \frac{1}{l}\sum_{j=1}^l E_{\rho_{s,j} }[\Pi_{s,j}^{(c)}] |\geq \varepsilon/2\right)\\\nonumber
& \leq 2 e^{-l |\gamma| \varepsilon^2/12(e^{|\gamma|}-1)}.
\end{align}
When $\left|\mu_c-\lambda_c\right| \leq \varepsilon/2$, the probability of the check passing is ($\leq 2 e^{-l |\gamma| \varepsilon^2/12(e^{|\gamma|}-1)}$). This means the algorithm is unlikely to output such concepts, as they do not represent significant deviations from the threshold. Therefore, the algorithm should select concepts $c^*$ such that $\left|\mu_{c^*}-\lambda_{c^*}\right| \geq \varepsilon/2$.

The \cref{alg:learning_projector_valued} thus can be applied to obtain a good estimate if $|\max_c \hat{\mu}_{c}(\gamma) - \mu_{c^*}(\gamma) | \geq 2\varepsilon$ and $|\hat{\mu}_{c^*}(\gamma) - \mu_{c^*}(\gamma) | \geq \varepsilon/2$.

By an argument analogous to the error proof as in Lemma 2 in Ref. \cite{fanizza_learning_2024}, the probability of error is bounded by $p_{\text{err}}^{(2)} \leq T(0.97^k+2(k+1)e^{-l |\gamma| \varepsilon^2/12(e^{|\gamma|}-1)})$.

Finally, we take the length of a block as
\begin{align}
l & = O\left(\frac{1}{\varepsilon^2}\max(\frac{(e^{|\gamma|}-1)^2}{\gamma^2}\log (Tkm/\delta),\, (\log m+C_1)^2)\right),
\end{align}
and the sample complexity as 
\begin{align*}
n & = Tk \times O\left(\frac{1}{\varepsilon^2}\max(\frac{(e^{|\gamma|}-1)^2}{\gamma^2}\log (Tkm/\delta),\, (\log m+C_1)^2)\right)\\
& = \frac{1}{\varepsilon^2} \log \frac{1}{\delta} \log^2 \frac{1}{\varepsilon} \times O\left(\max \left( \frac{(e^{|\gamma|}-1)^2}{\gamma^2} \log \left(\frac{m}{\delta} \log \frac{1}{\delta} \log^2 \frac{1}{\varepsilon}\right),\, (\log m+C_1)^2 \right)\right),
\end{align*}
ensuring that the total error probability $p_{\text{err}}^{(1)} + p_{\text{err}}^{(2)} \leq \delta$. Furthermore, when $\gamma\rightarrow 0$, the sample complexity recovers to the complexity of the standard ERM algorithm in Fact \ref{fact:QERM}.
\end{proof}

The algorithm used in \cref{theorem:TERM_proj} is formulated as \cref{alg:learning_projector_valued} in \cref{app:algo}. The guarantees of this algorithm are given in Proposition~\ref{prop:conditions}.

\section{PAC Generalization Bound of TERM}\label{sec:pac_bounds}
Extensive analysis of the generalization error in Ref.~\cite{aminian2024generalizationerrortiltedempirical} motivated the introduction of the new measure, namely that of the tilted generalization error. The tilted generalization error is the difference between the tilted empirical risk on the training set and its expected value over the data distribution $R(f)$  relative to unknown concept $f$. Formally, the tilted generalization error is defined by Ref.~\cite{aminian2024generalizationerrortiltedempirical} as follows:
\be
\text{Tilted Generalization Error} := R(f) - \widetilde{R}_{\gamma}(f).
\ee
The parameter $\gamma$ controls the tilt value. By Lemma~\ref{lemma:tiltconverg}, for $\gamma \to 0$, it recovers the standard generalization error; for $\gamma \to \infty$, it approximates the behavior of the worst-case loss. Considerations of tilted losses in quantum settings, including the definitions provided in this paper, require adopting the tilted generalization error paradigm. These ideas are further motivated in the next section as future work suggestions. The following theorem discussed in this section corresponds to a measure of the generalization error with respect to our definition of QTERM. The convergence of the tilted empirical risk to the true risk is shown as follows. For reference, the PAC bound of ERM for learning concept classes via the covering number is shown in Refs.~\cite{anthony2009neural, wolf2023mathematical}.

\begin{theorem}[PAC Generalization Bound of TERM via Uniform Covering Numbers]\label{theorem:PAC}
Consider hypotheses $h$ in any concept class $\mathcal{F}$ and loss function taking value in $[0,1]$, for $\gamma \in (0,\varepsilon)$ and any probability distribution $D$, we denote the tilted empirical risk as $\widetilde{R}_{\gamma}(h)$ and the corresponding population risk $R(h)$ as $\mathbb{E}_{Z \sim D}[L (h, Z)]$. For any $\varepsilon > 0$, over samples $S$, we obtain that
\be
\Pr_{S \sim D^n}\left[\exists\,h\in \mathcal{F}: \left|R(h)-\widetilde{R}_{\gamma}(h)\right| \ge \varepsilon\right] \le 8\Gamma_{1}\left(2n,\frac{\varepsilon}{8},\mathcal{G}_{\mathcal{F},L,\gamma}\right) \exp\left(-\frac{n\varepsilon^2}{32 |e^\gamma -1|^2}\right),
\ee
where $\Gamma_{1}\left(2n,\frac{\varepsilon}{8},\mathcal{G}_{\mathcal{F},L,\gamma}\right)$
is the covering number of $\mathcal{G}_{\mathcal{F},L,\gamma}$ on a sample of size $2n$ within error $\frac{\varepsilon}{8}$.
\end{theorem}

\begin{proof}
We remind readers that the tilted empirical risk is defined as in Definition \ref{def:term}:
\be
\widetilde{R}_{\gamma}(h) :=
\frac{1}{\gamma} 
\log \left( \frac{1}{n} \sum_{i=1}^n \exp\left(\gamma\,L\left(h,Z_i\right)\right) \right),
\ee
where $\{Z_i\}_{i=1}^n$ are i.i.d. samples from an unknown distribution $D$, and $L(h,Z)\in[0,a]$ is a bounded loss. We wish to show that, with high probability, $\widetilde{R}_{\gamma}(h)$ is uniformly close (over all hypotheses $h$) to the true risk 
\be
R(h) = \mathbb{E}_{Z\sim D}\left[L(h,Z)\right].
\ee

For simplicity, we denote each data point by $Z_i = (X_i,Y_i)$. The proof will follow in three steps, using the same techniques as the proof of PAC bound of empirical risk, namely symmetrization, concentration, and lifting the bound from the cover to all $h\in \mathcal{F}$ in Ref. \cite{wolf2023mathematical,anthony2009neural}.

(i) Symmetrization.
Let $\mathcal{F}$ be a family of hypotheses $h:\mathcal{X}\to\widehat{\mathcal{Y}}$.  
Let $L:\mathcal{F}\times\mathcal{Z}\to[0,a]$ be a bounded loss function. We have an i.i.d. sample $S = \{Z_1,\dots,Z_n\}$ from $D$. Let $S' = \{Z'_1,\dots,Z'_n\}$ be an independent \textit{ghost sample}\footnote{A \textit{ghost sample} is another sequence of data in all identical to the training data $S$.} of the same size $n$, also drawn randomly from $D$. Define its tilted empirical risk
\be
  \widetilde{R}'_\gamma(h) = 
  \frac{1}{\gamma} \log \left(\frac{1}{n}\sum_{i=1}^n \exp\left(\gamma\,L(h,Z'_i)\right)\right).
\ee
and the corresponding empirical risk as $\hat{R}'(h) = \frac{1}{n} \sum _{i=1}^n L(h,Z'_i)$. By triangle inequality, if $|R(h) - \widetilde{R}_{\gamma}(h)| \ge \varepsilon$, in terms of indicator functions there exists
\be\label{Ineq:PAC1}
\mathbbm{1}\{|R(h) - \widetilde{R}_{\gamma}(h)| \ge \varepsilon\}
\mathbbm{1}\{\left|R(h) - \hat{R}'(h)\right|\le \frac{\varepsilon}{2}\} \mathbbm{1}\{|\widetilde{R}'_\gamma(h) - \hat{R}'(h)|\le\frac{ \varepsilon}{4}\} \le
\mathbbm{1}\{|\widetilde{R}'_\gamma(h) - \widetilde{R}_\gamma(h)|\ge\frac{ \varepsilon}{4}\}.
\ee
Taking the expectation value with respect to $S'$ in Eq. (\ref{Ineq:PAC1}), and using Hoeffding's inequality together with the assumption on $n$ that $n \ge 4c^2\varepsilon^{-2}\ln 2$, it leads to 
\be
E_{S'}\left[\mathbbm{1}\{|R(h) - \hat{R}'(h)| \le \frac{\varepsilon}{2}\}\right] \ge 1-2\exp[-\frac{\varepsilon^2n}{2c^2}]\ge \frac{1}{2}.
\ee
For the last term in Eq. (\ref{Ineq:PAC1}),
\be
E_{S'}\left[\mathbbm{1}\{|\widetilde{R}'_\gamma(h) - \widetilde{R}_\gamma(h)|\ge\frac{ \varepsilon}{4}\}\right] = \Pr_{S'}\left[\exists \,h\in \mathcal{F}: \left|\widetilde{R}'_\gamma(h) - \widetilde{R}_\gamma(h)\right|\ge\frac{ \varepsilon}{4}\right].
\ee
Similarly for the third term, under the assumption on tilting that $\gamma \in (0,\varepsilon)$ in \cref{theorem:TERM_proj}, for $0 < \delta \le \frac{1}{2}$, we use 
\be
\Pr_{S'} \left[\exists \,h\in \mathcal{F}: |\widetilde{R}'_\gamma(h) - \hat{R}'(h)|\le\frac{ \varepsilon}{4}\right] = E_{S'}\left[\mathbbm{1}\{|\widetilde{R}'_\gamma(h) - \hat{R}'(h)|\le\frac{ \varepsilon}{4}\}\right] \ge 1-\delta \ge \frac{1}{2}.
\ee
Insert the bounds above into Eq. (\ref{Ineq:PAC1}), we therefore obtain
\be
  \Pr_{S}\left[\exists\, h\in \mathcal{F}: |R(h) - \widetilde{R}_{\gamma}(h)| \ge \varepsilon\right] \le
  4\Pr_{S,S'}\left[\exists \,h\in \mathcal{F}: |\widetilde{R}'_\gamma(h) - \widetilde{R}_{\gamma}(h)| \ge\frac{\varepsilon}{4}\right].
\ee

(ii) Concentration. Let's merge the samples, 
\be
  T = \left(S,\,S'\right) = \left(Z_1,\dots,Z_n,Z'_1,\dots,Z'_n\right).
\ee
Hence $|T|=2n$. For $\gamma \in (0,\varepsilon)$, let us denote the ``exponential loss‐function classes'':
\be
  \mathcal{G}_{\mathcal{F},L,\gamma} =
  \left\{ g_{h}(z) = \exp\!\left(\gamma L(h,z)\right) \Bigm| h\in \mathcal{F}\right\}.
\ee
Suppose that $L(h,z)\in [0,a]$ with $g_h(z)\in[1,\,e^{\gamma a}]$. By definition, there exists a finite set of functions $\{\,g_{1},\dots,g_{M}\}\subset\mathcal{G}_{\mathcal{F},L,\gamma}$ such that:
For every $h\in\mathcal{F}$, there is some $j\in\{1,\dots,M\}$ with $\max_{z\in T}\,\left|g_h(z) - g_j(z)\right| \le \eta$, for some chosen $\eta>0$. The minimal such $M$ is 
$\Gamma_{1}\left(2n,\eta,\mathcal{G}_{\mathcal{F},L,\gamma}\right)$, the uniform covering number of $\mathcal{G}_{\mathcal{F},L,\gamma}$ on a sample of size $2n$.  

Suppose we fix some $g_j$. Then
\be
  \widehat{\Phi}(g_j) 
   := \frac{1}{n}\sum_{i=1}^n g_j(Z_i),
  \quad
  \widehat{\Phi}'(g_j) 
   := \frac{1}{n}\sum_{i=1}^n g_j(Z'_i).
\ee
Since $g_j(z)\in [1,e^{\gamma a}]$, we define $d(\gamma,a) = |1- e^{\gamma a}|$ and apply Hoeffding’s inequality to obtain
\be
  \Pr\left[
    \left|\widehat{\Phi}'(g_j)-\widehat{\Phi}(g_j)\right|
     \ge t
  \right]
   \le 2 \exp \left(- \frac{n t^2}{2 d^2(\gamma,c)}\right).
\ee
Applying a union bound over $j=1,\dots,M$, we get
\be
  \Pr\left[
    \exists  j: 
    \left|\widehat{\Phi}'(g_j) -\widehat{\Phi}(g_j)\right|\ge t
  \right]\le
  M\cdot 2\exp\left(-\frac{n t^2}{2 d^2(\gamma,c)}\right).
\ee
Hence, choosing $t$ suitably as $\frac{\varepsilon}{4}$, with probability at least 
\be
  1 - 2\Gamma_{1}\left(2n,\eta,\mathcal{G}_{\mathcal{F},L,\gamma}\right)
           \exp\left(- \frac{n\varepsilon^2}{32 d^2(\gamma,c)}\right),
\ee
all the elements $\{g_j\}$ have
$\left|\log(\widehat{\Phi}'(g_j)) - \log(\widehat{\Phi}(g_j))\right|\le \frac{|\widehat{\Phi}'(g_j) -\widehat{\Phi}(g_j)|}{\min(1, e^{\gamma a})} \le  \frac{\varepsilon}{4} $ if $a=1$.

(iii) Lifting the bound from the cover to all $h\in \mathcal{F}$. Now, take any $h\in\mathcal{F}$. By construction, it is $\eta$‐close (in $\ell_\infty$ on sample $T$) to one of the cover elements, say $g_j$. Specifically, 
\be
  \max_{Z_i \in T} \left|\widehat{\Phi}(h) - \widehat{\Phi}(g_j)\right| = \max_{Z_i \in T} 
  \left|
    \frac{1}{n}\sum_{i=1}^n \left(e^{\gamma L(h,Z_i)}-g_j(Z_i)\right)
  \right|\le\eta,
\ee
and similarly for the ghost sample. But we must ensure every $h\in\mathcal{F}$ cannot exceed $\frac{\varepsilon}{2}$, suppose that by decomposition
\be
\left|\log(\widehat{\Phi}'(h)) - \log(\widehat{\Phi}(h))\right|
\le  \left|\log(\widehat{\Phi}'(h)) - \log(\widehat{\Phi}'(g_j))\right| +
\left|\log(\widehat{\Phi}'(g_j)) - \log(\widehat{\Phi}(g_j))\right| + \left|\log(\widehat{\Phi}(g_j)) - \log(\widehat{\Phi}(h))\right|.
\ee
Indeed, $\widehat{\Phi}(h)$ and $\widehat{\Phi}(g_j) $ differ by at most  $\eta $ in average, so their logs also differ by  $\le \eta$ when $c=1$ for chosen $\gamma$, which we choose to be  $\le \varepsilon/8 $. Similarly for $\widehat{\Phi}'(h) $ or $\widehat{\Phi}'(g_j) $. Hence, no such $g_j$ in the middle term can exceed $\frac{\varepsilon}{4}$ with high probability derived in concentration step. 

Finally, putting it all together, we obtain the desired uniform bound:
\be
  \Pr_{S \sim D^n}\left[\exists\,h\in \mathcal{F}: \left|R(h)-\widetilde{R}_{\gamma}(h)\right| \ge \varepsilon\right] \le 8\Gamma_{1}\left(2n,\frac{\varepsilon}{8},\mathcal{G}_{\mathcal{F},L,\gamma}\right) \exp\left(-\frac{n\varepsilon^2}{32 |e^\gamma -1|^2}\right).
\ee
This completes the proof of the PAC‐type (uniform convergence) guarantee for the tilted empirical risk.
\end{proof}

\begin{remark}
When hypotheses are implemented via parametrized quantum circuits with specific data-encoding strategies, one can leverage encoding-dependent generalization bounds in Ref.~\cite{caro2021encodingdependent} to obtain explicit estimates of the covering numbers for the tilted-loss class under that encoding.
\end{remark}
In quantum learning settings, it is common to bound the function-class covering number $\Gamma_{1}(n,\varepsilon,\mathcal{G}_{\mathcal{F},L,\gamma})$ by the operator-class covering number  $\Gamma_{1,q}(n,\varepsilon,\mathcal{C})$. This arises since each real-valued function in  $\mathcal{G}_{\mathcal{F},L,\gamma} $ originates from measuring operators in $\mathcal{C} $, thus yielding
$\Gamma_{1}(n,\varepsilon,\mathcal{G}_{\mathcal{F},L,\gamma}) \le
\Gamma_{1,q}(n,\varepsilon,\mathcal{C})$. Therefore, the function-class covering number appearing in the PAC bound established in \cref{theorem:PAC} can generally be bounded by the operator-class covering number, thus facilitating concrete bounds on the learnability of quantum hypothesis classes via TERM. In the next section, we will define uniform convergence of QTERM for agnostic learnability bounds.

\section{Statistical Learning for Projector-Valued Functions}\label{sec:agnostic_learning}
In this section, we present a lemma which ensures that optimizing the TERM over $\varepsilon$-net of the concept class depending on classical data, yields a solution that closely approximates the true risk minimizer.

\begin{lemma}[Uniform Convergence of QTERM]\label{lemma:UnifCov}
Given $l_0$ samples drawn from the classical-quantum distribution 
\be
\rho = \sum_{x \in \mathcal{X}} D(x)|x\rangle\langle x| \otimes \rho(x),
\ee  
one can construct a finite subset $\mathcal{C}_{\varepsilon/2} \subseteq \mathcal{C}$ with respect to the loss function $L$ whose size obeys $|\mathcal{C}_{\varepsilon/2}| \le \Gamma_{1,q}(l_0,\varepsilon/2,\mathcal{C})$. For each hypothesis $ h \in \mathcal{C} $, there exists a known $ h^* \in \mathcal{C}_{\varepsilon/2} $ such that 
\be
|R(h)-R(h^*)|\leq \varepsilon.
\ee  
Furthermore, with high probability, the tilted empirical risk uniformly approximates the true risk
\be
\forall h \in \mathcal{C}: | R(h)-\widetilde{R}_{\gamma}(h)|<\varepsilon/4,
\ee  
with probability of error at most 
\be
p_{\text{err, unif}} \leq 8\Gamma_{1,q}(2l_0,\varepsilon/128,\mathcal{C})e^{-\frac{l_0\varepsilon^2}{512|e^{\gamma}-1|^2}}.
\ee
\end{lemma}
\begin{proof}
Apply the PAC bound of TERM derived in \cref{theorem:PAC} to a sample $S$ of size $l_0$, where each sample is drawn independently from distribution $D$. With probability at least $1-p_{\text{err, unif}}$, we have that for tilted risk $\widetilde{R}_{\gamma}(h)$, 
\be
\Pr_{S\sim D}\left[\exists h \in \mathcal{C}: |R(h)-\widetilde{R}_{\gamma}(h)|\geq \varepsilon/4\right]\leq 8\Gamma_{1,q}(2l_0,\varepsilon/128,\mathcal{C})e^{-\frac{l_0\varepsilon^2}{512|e^{\gamma}-1|^2}}.
\ee  
where $R(h)$ denotes the true risk under the chosen loss function. Construct an $ \varepsilon/2 $-net $ \mathcal{C}_{\varepsilon/2}(\vec{x}) $ with respect to the pseudo-metric $d_r(g_1,g_2) =|\widetilde{R}_\gamma(g_1) - \widetilde{R}_\gamma(g_2) |$ on the data $ \vec{x} $. By definition of net, for every $h \in \mathcal{C}$ there is $h^* \in \mathcal{C}_{\varepsilon/2}(\vec{x}) $ such that  
\be
|\widetilde{R}_\gamma(h)-\widetilde{R}_{\gamma}(h^*)|\leq \varepsilon/2.
\ee
With the assumptions from \cref{theorem:PAC} in hand, for any $h \in \mathcal{C}$, with high probability, we have  
\begin{align}
|R(h)-R(h^*)|
&\leq |R(h)-\widetilde{R}_{\gamma}(h)|+|\widetilde{R}_{\gamma}(h)-\widetilde{R}_{\gamma}(h^*)|+|R(h^*)-\widetilde{R}_{\gamma}(h^*)|  \\\nonumber
&\leq 2\sup_{h \in \mathcal{C}}|R(h)-\widetilde{R}_{\gamma}(h)|+\varepsilon/2  
\\\nonumber
&\leq \varepsilon.
\end{align}
\end{proof}
\begin{remark}
The choice of loss function determines which Schatten-$q$ seminorm (e.g., $q=1$ for trace‐distance losses, $q=\infty$ for operator‐norm losses) is used to define the pseudometric between two concepts.  
\end{remark}
Via this lemma, we decompose the learning task into two stages: using classical data to construct an $\varepsilon$-net over the concept space, and then using quantum data to select an optimal hypothesis from this finite subset. The $\varepsilon$-net ensures coverage of the true risk landscape, while the quantum component enables effective hypothesis selection under the TERM framework. This decomposition provides a practical and theoretically grounded approach to regularized learning from quantum data. We now formalize these in the following theorem on the agnostic learnability of QTERM framework.

\begin{theorem}[Agnostic Learnability of Projector-Valued Functions via QTERM]\label{thm:agnostic_QTERM}
Let $\mathcal{C}$ be a concept class consisting of a set of projectors $h_c:=\{ \Pi_1^{(c)}, \cdots, \Pi_n^{(c)} \}_{c=1}^m$ and let $\varepsilon >0$. Given training samples $ S = \{(x_i, \rho_i )\}_{i=1}^n$ with $x_i \sim D$ and $\rho_i$ produced by an unknown classical-quantum channel, there exists an agnostic learning algorithm $\mathcal{A}$ that, for fixed $\gamma$ in a proven range, outputs a hypothesis $c^* \in \mathcal{C}$ along with an estimate $\hat{\mu}_{c^*}(\gamma)$ of $1 - \widehat{\widetilde{R}}_{\gamma}(c)$ such that
\be
\Pr\left(\left|\hat{\mu}_{c^*}(\gamma) - \inf_{c \in \mathcal{C}} \mu_c(\gamma) \right| \geq 3\varepsilon \ \cup \ \left|\hat{\mu}_{c^*}(\gamma) - \mu_{c^*}(\gamma) \right| \geq 2\varepsilon \right) := p_{\text{err}},
\ee
if the covering numbers satisfy
\be
\lim_{n\to \infty} \frac{\log^2 \Gamma_{1,\infty}(n, \varepsilon, \mathcal{C})}{n} = 0, \quad \forall \varepsilon >0.
\ee
And $p_{\text{err}}$ can be made arbitrarily small provided that $n = 6Tkl$, and $T=O(\log\frac{1}{\varepsilon})$ and $k=O(\log\frac{1}{\delta}\log\frac{1}{\varepsilon})$, for constants $C_1,\,C_2,\,C_3$ and $l$ such that
\be
(\log \Gamma_{1,\infty} (6Tkl,\varepsilon/2,\mathcal{C})) +C_2)^2 \le C_1 l\varepsilon^2,
\ee
we have
\be
p_{\text{err}} \le \frac{\delta}{2} + C_3 Tk \Gamma_{1,\infty} (6Tkl,\varepsilon/2,\mathcal{C}) e^{-\frac{l |\gamma| \varepsilon^2}{12(e^{|\gamma|}-1)}}) + 8\Gamma_{1,\infty} (12Tkl,\varepsilon/128,\mathcal{C})e^{-\frac{6Tkl\varepsilon^2}{512|e^\gamma -1|}}.
\ee
\end{theorem}

\begin{proof}
The failure probability $p_{\text{err}}$ is decomposed into two independent contributions, each bounded by $\delta/2$ such that
\be
p_{\text{err}}\le p_{\text{err,unif}} + p_{\text{err,TERM}}\le \delta.
\ee
The proof proceeds in two main steps. First, we bound the uniform convergence of the tilted empirical risk, and then we show that the TERM algorithm identifies a near-optimal hypothesis over a suitably chosen $\varepsilon/4$-net of $\mathcal{C}$.

(i) Uniform Convergence. Noted that the tilted empirical risk for projector-valued functions
\be
\widetilde{R}_{\gamma}(h) := 1-\mu_c(\gamma) = 1- \frac{1}{\gamma} \log \left( \frac{1}{n} \sum_{i=1}^n e^{\gamma \mathrm{Tr}[\rho_i \Pi_i^{(c)}]} \right).
\ee
Under some assumptions (see \cref{theorem:PAC}), it yield that the probability of error is bounded by \cref{lemma:UnifCov},
\be
p_{\text{err, unif}}(6Tkl) \leq 8\Gamma_{1,\infty} (12Tkl, \varepsilon/128, \mathcal{C})e^{- \frac{6Tkl\varepsilon^2 }{512|e^\gamma -1|^2}}
\ee  

(ii) TERM algorithm on $\varepsilon/2$-net of $\mathcal{C}$. Using the classical samples $\vec{x}$ to construct an $\varepsilon/2$-net $\mathcal{C}_{\varepsilon/2}(\vec{x})$ for $\mathcal{C}$. Apply \cref{theorem:TERM_proj} (ThresholdSearch + Check) to the net of size $m = \Gamma_{1,\infty} (6Tkl, \varepsilon/2, \mathcal{C})$, and output a $c^*$ satisfying 
$$
|\hat{\mu}_{c^*}(\gamma) - \max_{c \in \mathcal{C}(\Vec{x})} \mu_{c}(\gamma)| \leq 2\varepsilon,\;\text{and}\;
|\hat{\mu}_{c^*}(\gamma) - \mu_{c^*}(\gamma)| \leq 2\varepsilon$$
such that
\be
\hat{\mu}_{c^*}(\gamma) < \max_{c \in \mathcal{C}(\vec{x})} \mu_c(\gamma) + 2\varepsilon
\ee
The probability of error of the TERM is bounded as in the proof of \cref{theorem:TERM_proj}, 
\be
p_{\text{err, TERM}} \leq T(0.97^k+2(k+1)e^{-\frac{l |\gamma| \varepsilon^2}{12(e^{|\gamma|}-1)}}) + 4Tk \Gamma_{1,\infty} (6Tkl, \varepsilon/2, \mathcal{C}) e^{-\frac{l \gamma^2 \varepsilon^2}{32(e^{|\gamma|} -1)^2}},
\ee
by choosing $T=O(\log\frac{1}{\varepsilon})$, $k=O(\log\frac{1}{\delta}\log\frac{1}{\varepsilon})$ and $l$ so that $(\log \Gamma_{1,\infty} (6Tkl,\varepsilon/2,\mathcal{C})) +C_2)^2 \le C_1 l\varepsilon^2$. Thus, for the selected concept $c^* \in \mathcal{C}_{\varepsilon/2}(\vec{x})$, we have
\begin{align}
\mu_{c^*}(\gamma) &\leq \hat{\mu}_{c^*}(\gamma) + \frac{\varepsilon}{4}
\leq \inf_{c \in \mathcal{C}_{\varepsilon/2}(\Vec{x})} \hat{\mu}_{c}(\gamma) + 2\varepsilon + \frac{\varepsilon}{4}\leq \inf_{c \in \mathcal{C}} \hat{\mu}_{c}(\gamma) + 2\varepsilon + \frac{3\varepsilon}{4} \leq \inf_{c \in \mathcal{C}} \hat{\mu}_{c}(\gamma) + 3\varepsilon
\end{align}
Hence, the probability of error is then upper bounded by both steps, $p_{\text{err, unif}} + p_{\text{err, TERM}}$ as the desired result. 
\end{proof}
In essence, the subexponential growth of $\Gamma_{1,\infty} (n, \varepsilon, \mathcal{C})$ guarantees that the TERM algorithm identifies a projector‐valued hypothesis $c^*$ whose \emph{tilted empirical risk} is within $3\varepsilon$ of the best possible among all hypotheses in $\mathcal{C}$. Thus, we achieve agnostic learnability of such projector‐valued functions under the stated covering number condition.

\section{Quantum Tilted Risk (QTR)}\label{sec:qtr}
Tilted risk measures extend traditional risk evaluation by incorporating a parameter to control the influence of extreme outcomes. In the quantum setting, this concept adapts to the structure of quantum states and observables. Inspired by classical notions such as the Esscher transform \cite{esscher1932probability, gerber1994option} and tilted risk measures, we define the Quantum Tilted Risk (QTR), which generalizes these concepts to quantum settings (we refer the reader to Ref.~\cite{qiu_quantum_2024} for quantum notions on the Esscher transform). The formulation involves a parameterized Hamiltonian $H(\theta)$ with an input quantum state $\rho$, allowing flexible adjustments via the hyperparameter $\gamma$. The example of Hamiltonian learning \cite{bakshi2023learningquantumhamiltonianstemperature, Anshu_2021, gu2024practical, Yu_2023}, motivates defining the measure of QTR.

\begin{definition}[Quantum Tilted Risk (QTR)]
For all $\gamma \in \mathbb{R}$, and all $\phi \in \Phi$, given an input quantum state $\rho$ and an Hamiltonian $H(\phi)$ parameterized by $\phi$, the $\gamma$-tilted QTERM is defined as:
\be
\widetilde{R}_Q(\gamma,\phi) := \frac{1}{\gamma}\log \tr \left( e^{\gamma H(\phi)} \rho \right).
\ee
\end{definition}

In a learning setting, if $\rho$ represents data encoded in a quantum state and $H(\phi)$ represents a hypothesis or model parameter $\phi$, then $\tr \left( e^{\gamma H(\phi)} \rho \right)$ measures how well the model $H(\phi)$ aligns with the data distribution $\rho$ in an exponential tilted scenario. The parameter $\gamma$ could be tuned to focus on tail events or rare configurations emphasized by $H(\phi)$.

\subsection{Properties of QTR}\label{sec:QTR_prop}
\begin{lemma}[Convergence of QTR]\label{Lemma:convergence}The convergence of QTR to the expectation value of Hamiltonian is
$$\lim _{\gamma \rightarrow 0} \widetilde{R}_Q(\gamma , \phi) := \widetilde{R}_Q(0 , \phi) = \tr \left( H(\phi) \rho \right).$$

\begin{proof}
\be
\lim _{\gamma \rightarrow 0} \widetilde{R}(\gamma , \phi) & &=\lim _{\gamma \rightarrow 0} \frac{1}{\gamma}\log \tr \left( e^{\gamma H(\phi)} \rho\right) \\\nonumber
(\text{L’Hôpital’s rule applied to $\gamma$}) && =\lim _{\gamma \rightarrow 0} \frac{ \tr \left( H(\phi) e^{\gamma H(\phi)} \rho\right)}{\tr \left( e^{\gamma H(\phi)} \rho\right)} \\\nonumber
&& =\tr \left( H(\phi) \rho \right).
\ee
\end{proof}
\end{lemma}

\textbf{Relation to quantum Rényi relative entropy.} 
The quantum Rényi relative entropy of order $\alpha>0 \neq 1$ is defined as
$$
D_\alpha(\rho \|\sigma) := \frac{1}{\alpha-1} \log [\tr (\rho^\alpha \sigma^{1-\alpha})].
$$
As $\alpha \rightarrow 1$, it recovers to standard quantum relative entropy. Let $\sigma(\phi)$ be a quantum state parameterized by $\phi$. Define $H(\phi)=\log \sigma(\phi)$ such that
$$
\tr \left( e^{\gamma \log \sigma(\phi)} \rho \right) = \tr \left( \sigma(\phi)^{\gamma} \rho \right).
$$
$\widetilde{R}_Q(\gamma,\phi)$ can be seen as a \textit{tilted} relative Rényi entropy to measure how $\rho$ and $\sigma(\phi)$ overlap. 

\textbf{Operational Interpretation.} When $H(\phi)$ is viewed as a Hamiltonian (i.e., an energy operator) and $\rho$ as a quantum state (i.e., a density matrix at fixed temperature), then
$$
Z_\phi(\gamma):= \tr \left( e^{\gamma H(\phi)} \rho \right).
$$
resembles a ``partition function'' for the state $\rho$ under the energy landscape defined by $H(\phi)$. The tilted quantity 
$$
\widetilde{R}_Q(\gamma,\phi) := \frac{1}{\gamma}\log \tr \left( e^{\gamma H(\phi)} \rho \right) = \frac{1}{\gamma}\log Z_\phi(\gamma),
$$ 
resembles a quantity of free energy ($F=-k_B T \ln Z$) at the inverse temperature. 

\subsection{On Different Methods of Formulating Tilts}
The idea behind implementations of the tilted hyperparameter warrants physical differences among different formulations. Here we discussed one main definition formalized as QTERM relevant in learning quantum processes. We also briefly motivated the need for tilted measures in other setting such as Hamiltonian learning in a measures denoted as QTR. Below we compare the scenarios where either may be relevant.

\begin{itemize}
    \item \textbf{QTERM.} Specifically tailored for empirical risk minimization in a learning setting. It evaluates the performance of quantum hypotheses on given data, with the tilting parameter $\gamma$ allowing emphasis on particular outcomes (i.e., high-loss scenarios). It is defined over discrete samples and focuses on empirical data.
    \item \textbf{QTR.} A generalized risk measure that quantifies the alignment of a parameterized quantum model (Hamiltonian $H(\phi)$) with a quantum state $\rho$ and any state learning problem. It is more theoretical and continuous, focusing on the expected performance in a weighted (tilted) sense.
\end{itemize}

One of the main differences is found in the data type they are designed to process. In QTERM quantum data is represented as a set of quantum projectors ${\Pi_i^{(c)}}$ derived from empirical observations. While in QTR, quantum data is encoded in a density operator $\rho$, representing a general quantum state. Furthermore, while QTERM focuses on finite-sample empirical data in a learning framework, QTR is concerned with global system-level evaluations involving quantum states and observables. The distinct definitions reflect the different goals: QTERM emphasizes empirical data alignment and hypothesis testing, whereas QTR captures broader properties of quantum systems with a tilt to explore rare or significant states. One can expect other definitions to arise as the scope of quantum machine learning broadens, particularly to address scenarios such as multi-objective optimization in quantum systems, settings where data distributions exhibit extreme variations, or where hybrid quantum-classical models require tailored risk measures. For instance, tilts might be redefined to handle non-Hermitian operators, time-evolved quantum states, or scenarios emphasizing specific quantum observables, such as entanglement or coherence. We leave the above as future work suggestions.

\section{Concluding Thoughts and Open Directions}\label{sec:conclusion}
This work introduces QTERM, a quantum extension of tilted empirical risk minimization that applies a tilt hyperparameter to quantum empirical risk. It unifies and extends classical TERM and QERM, offering a new approach tailored to the challenges of quantum data and learning. Inspired by the classical setting, QTERM introduces a tilted hyperparameter that penalizes outliers or class imbalances in datasets. A central idea behind this tilt mechanism is to help prevent the model from fitting excessively to noisy or anomalous data points, which could lead to overfitting. Further, we have developed theoretical bounds on sample complexity and learning guarantees in the context of quantum learning with QTERM. Additionally, we have proved the PAC-type generalization bound for the classical TERM. These two findings were used to prove an upper bound on the agnostic learnability of QTERM. Our results on quantum process learning provide an understanding of how tilted losses influence the scaling of quantum learning tasks with respect to the number of quantum samples. This work contributes to the broader field of quantum learning by offering a theoretical understanding for implicitly regularized quantum models, making quantum learning more practical and scalable.

A natural extension of this work would be to expand theoretical proofs to cover larger ranges of the tilt $\gamma$. While our current analysis focuses on the small-$\gamma$ regime to ensure the provable sample complexity bound and learnability, understanding the theoretical behavior of TERM for large $\gamma$ remains an open challenge. Such an extension could reveal new insights into the robustness of the method under stronger regularization and its applicability to learning scenarios with highly irregular or noisy quantum data.

From the main perspective, QTERM is an advanced quantum loss function that can be practically relevant to quantum learning models. While the study of loss functions and quantum regularization has already made efforts in understanding certain types of regularizers, i.e., in establishing lower bounds with respect to Lasso and Ridge regression \cite{chen2022quantumalgorithmslowerbounds}, and quantum regularized least squares \cite{Zhu2024connectionbetween, Chakraborty_2023} --- there remains ample opportunity to further study QTERM, as well as other ways of regularizing. In particular, future research could focus on comparing and establishing properties of QTERM with respect to other types of regularizers such as $L_2$ regularization and $L_1$ regularization, and adding regularizers such as entropy to QTERM as well \cite{aminian2025generalization}, for various quantum learning settings, including quantum circuit learning, quantum neural networks, and quantum reinforcement learning, to determine how QTERM influences model performance, convergence, and generalization in, for example, noisy quantum environments. Additionally, as a future work direction, we suggest further generalization bounds which can be derived for QTERM, as a well as other properties such as robustness measures and convergence rates, similar to the methodologies in Ref.~\cite{aminian2024generalizationerrortiltedempirical}. In terms of concrete learning model examples, QTERM could be extended to quantum generative models, such as quantum GANs or variational autoencoders, and quantum reinforcement learning.

In classical optimization, the weighted loss corresponds to the value of the tilt applied to the model's parameters. This approach is useful in scenarios where different parts of the model need varying levels of emphasis. To address this, the classical gradient calculation can be adjusted to incorporate the tilted hyperparameter, allowing the gradient to reflect the relative importance of each parameter. In quantum optimization, gradients are similarly required for quantum circuits. The definition of QTERM proposed here can be incorporated into existing quantum training algorithms, such as the parameter-shift rule, which is a method for evaluating the gradients of quantum observables with respect to the parameters of quantum circuits \cite{mitarai2018quantum, Schuld_2019, wierichs2022general}. In training quantum machine learning models, the parameter-shift rule addresses the challenges of applying the backpropagation algorithm \cite{Rumelhart1986LearningRB} to these quantum models \cite{abbas2023quantumbackpropagationinformationreuse}.
In context with this work, presence of the tilting hyperparameter motivates us to define tilted parameter shift rule, which adapts the parameter shift rule to handle weighted losses, offering a more tailored approach to optimization in quantum algorithms.

With respect to other characterizations and generalization behavior for TERM applied to quantum data, several avenues are worth exploring. For instance, margin bounds could be relevant here, as recent work shows that margin-based approaches can provide insightful guarantees about generalization \cite{hur2024understandinggeneralizationquantummachine, neyshabur2018pacbayesianapproachspectrallynormalizedmargin, hanneke2020stablesamplecompressionschemes}. In the case of TERM, the value of the tilted hyperparameter influences the decision boundary, which can be a motivating point to show how larger margins reduce the bound on generalization error. This behavior suggests that the tilted hyperparameter may implicitly control the effective margin of the classifier. Larger margins are often associated with tighter generalization bounds, providing a compelling motivation to investigate how varying the tilted hyperparameter could directly influence the margin properties and, consequently, reduce the bound on generalization error. Similarly, stability-based bounds for generalization error can be a compelling framework \cite{bousquet2002stability}. Stability measures the sensitivity of the model to changes in the training set, which directly impacts its generalization ability.

In summary, in this work we define quantum tilted empirical risk, and in turn extend the classical principle of TERM to quantum learning systems, positioning it as a generalization technique. The QTERM framework we propose here, derived from Ref.~\cite{fanizza_learning_2024} incorporates a tilted hyperparameter $\gamma$ into the loss function. Through this framework, we derived three central theoretical results: a sample complexity guarantee for learning projector-valued functions under QTERM, a PAC-style generalization bound for tilted loss minimization in the small-$\gamma$ regime, and an agnostic learning guarantee that ensures reliable hypothesis selection even in the absence of a perfect model. Together, these results establish a comprehensive foundation for understanding the statistical and algorithmic properties of learning in quantum systems using tilted losses. We believe that QTERM provides a promising direction for future work in quantum learning, especially in exploring adaptive strategies for tuning $\gamma$ and extending the analysis to broader classes of quantum models.

At a broader level, this work highlights the continuing dialogue between classical and quantum perspectives on learning. Classical learning theory has provided a foundation for understanding generalization, regularization, and the computational limits of artificial systems. Extending these ideas into the quantum domain through QTERM reveals how quantum principles reshape long-standing questions of efficiency, robustness, and adaptability in learning. In this sense, QTERM forms part of a larger effort to develop a coherent theory of quantum learning that is both continuous with and distinct from its classical counterpart. In doing so, the study of quantum learning moves beyond algorithms to touch on deeper questions about the nature of information, computation, and intelligence itself.

\vspace{0.4cm}
\section{Acknowledgments}
We thank Matthias C. Caro, Vedran Dunjko and Jens Eisert for discussions and feedback.
This work is supported by the National Research Foundation, Singapore, and A*STAR under its CQT Bridging Grant and its Quantum Engineering Programme under grant NRF2021-QEP2-02-P05.
Portions of this manuscript were drafted or edited with the assistance of ChatGPT to improve clarity and style.

\section{Appendix}
\subsection{Lemmas and Propositions}\label{app:proofs}
\begin{lemma}[Multiplicative Chernoff bound]\label{lemma:MCbound}
Let $X_1, \cdots, X_n$ be independent random variables taking values in $[0, 1] $. Define $X:=\sum_{i=1}^n X_i$ and let $E[X]$ be the expected value of $X$. For any $0 \leq \varepsilon \leq 1$,
\begin{align}
\Pr\left( |X- E[X] | \geq n\varepsilon \right) \leq 2e^{-E[X]\varepsilon^2/3}.
\end{align}
\end{lemma}

\begin{proposition}[Multiplicative Chernoff bound for general bounded variables]\label{Prop:Chernoff}
Let $X_1, \cdots, X_n$ be independent random variables taking values in $[a, b]$ that satisfy $a, b\geq 0$. Define $X:=\sum_{i=1}^n X_i$ and let $E[X]$ be the expected value of $X$. For any $0<\delta<1$,
\begin{align}
\Pr(|X - E[X]| \geq n\delta) \leq 2e^{-n\delta^2/[3(b-a)^2]}.
\end{align}
\end{proposition}
\begin{proof}
    When random variables are not restricted to $[0,1]$ but instead lie in the positive interval $[a,b]$, we can reduce to the case $[0,1]$-by normalization. Let $X_1, \ldots, X_n \in [a,b]$ be independent random variables. Define:
$$
Y_i = \frac{X_i - a}{b-a},\quad Y := \sum_{i=1}^n Y_i
$$
so that $Y_i \in [0,1]$ and
$$
\mu_Y := \mathbb{E}[Y] = \frac{\mathbb{E}[X] - a n}{b-a} \in [0,n], \;\text{with} \; X = \sum_{i=1}^n X_i = a n + (b-a)Y.
$$
By the standard multiplicative Chernoff bound for $Y$, for any $\varepsilon>0$ we have
$$
Pr (Y \geq (1+\varepsilon)\mu_Y) \le \exp\left(-\frac{\mu_Y \varepsilon^2}{3}\right)
$$
and similarly for the lower tail. In terms of $X$, note that
$$
Y \geq (1+\varepsilon)\mu_Y \implies X \ge a n + (b-a)(1+\varepsilon)\mu_Y = \mu_X (1 + \varepsilon) - a n\varepsilon,
\\
\Pr(X \geq \mu_X (1 + \varepsilon) - a n\varepsilon) \leq \exp\left(-\frac{(\mu_X - a n)\varepsilon^2}{3(b-a)}\right),
$$
where $\mu_X =E[X]$. Similarly, it obtains
$$
\Pr(X \geq \mu_X (1 - \varepsilon) + a n\varepsilon) \leq \exp\left(-\frac{(\mu_X - a n)\varepsilon^2}{2(b-a)}\right).
$$
Combine the above bounds to obtain
\begin{align*}
\Pr(|X - \mu_X| \geq (\mu_X - a n)\varepsilon) \leq 2\exp\left(-\frac{(\mu_X - a n)\varepsilon^2}{3(b-a)}\right).
\end{align*}
Setting $(\mu_X-an)\varepsilon = n\delta$ and using the fact that $\mu_X-an \le n(b-a)$ leads to
\begin{align*}
\Pr(|X - \mu_X| \geq n\delta) \leq 2\exp\left(-\frac{n\delta^2}{3(b-a)^2}\right)
\end{align*}
If $a=0$ and $b=1$, it recovers the standard multiplicative Chernoff bound. 
\end{proof}

\begin{lemma}[Naive expectation estimation, Proposition 1 in \cite{fanizza_learning_2024}]\label{lemma:NEestimation}
The random variable $X$ obtained by measuring the observable $\Bar{\Pi}:= \Pi_1\otimes\mathbf{1}\otimes\cdots\otimes\mathbf{1} + \mathbf{1}\otimes\Pi_2\otimes\mathbf{1}\cdots\otimes\mathbf{1} + \mathbf{1}\otimes\cdots\otimes\mathbf{1}\otimes\Pi_n$ on a quantum state $\rho=\rho_1\otimes \cdots \otimes \rho_n$ satisfies, for $\varepsilon<1$,
\begin{align}
Pr\left( \left|X- E_{\rho}[\Bar{\Pi}] \right| \geq n\varepsilon\right) \leq 2e^{-n\varepsilon^2/3}.
\end{align}
\end{lemma}
Note that $E_{\rho}[\Bar{\Pi}]=\sum_{s=1}^n E_{\rho_s}[\Pi_s]:= n(1-R(h))$ with empirical risk of hypothesis $R(h)$. Based on Lemma~\ref{lemma:NEestimation}, we extend the concentration result to exponentiated observables, stated next as Proposition~\ref{lemma:NEestimationExp}.

\begin{proposition}[Exponential expectation estimation]\label{lemma:NEestimationExp}
Let the random variable $\{Y_s\}_{s\in [n]}$ be the measurement outcomes by measuring $\bar{\Pi} = \{ \Pi_s\}_{s\in [n]}$ on a collection of a product state $\rho$ components $\{\rho_s\}_{s\in [n]}$. Define $X_s:= e^{c Y_s}$ with a constant $c\in \mathbb{R}$. Since $Y_s \in [0,1]$, each \(X_s\) lies in the interval $[\min\{1,e^c\}, \max\{1,e^c\}]$ and the length of the interval of $r =|e^c-1|$. And define $X:=\sum_{s=1}^n X_s=\sum_{s=1}^n e^{c Y_s}$. Thus, for any $0\leq \varepsilon \leq 1$,
\begin{align}\label{Ineq:expon_estimation}
\Pr\left( \left|X- E_{\rho}[X(\Bar{\Pi})] \right| \geq n\varepsilon\right) \leq 2e^{-n\varepsilon^2/3r^2},
\end{align}
where $E_{\rho}[X(\Bar{\Pi})]=\sum_{s=1}^n e^{c E_{\rho_s}[\Pi_s]}$ represents is the sum of the individual expectation values of the exponentiated outcomes.
\end{proposition}
\begin{proof} 
The measurement outcomes $\{Y_s\}_{s\in [n]}$ are independent scalar random variables. Consequently, $\{X_s\}_{s\in [n]}$ are independent and each is bounded in the interval. By applying multiplicative Chernoff bound in Proposition \ref{Prop:Chernoff} and let $r^2=|1-e^c|^2$, one obtains the result in \eqref{Ineq:expon_estimation}.
\end{proof}

\begin{lemma}[Hoeffding's Inequality] \label{lemma:Hoeffding}
Let $\mathcal{X} = x_1,\cdots, x_N$ be a finite population of
points, each taking values in an interval $[a,b]$. Suppose $\{X_1, \cdots, X_n\}$ is a random sample drawn without replacement from $\mathcal{X}$. Let $X:=\sum_{i=1}^n X_i$ and let $\mu$ denote the empirical mean of $\mathcal{X}$. Then for all $t>0$,
$$
Pr\left( \left|\sum_{i=1}^n X_i - n \mu \right| \geq n t\right) \leq 2e^{-2n t^2/(a-b)^2}.
$$
\end{lemma}

Let us extend Hoeffding's inequality for sampling without replacement to the case of $M$ different finite populations. Each population has $N$ points and we draw $n$ samples without replacement from each population. 

\begin{proposition}[Hoeffding’s Inequality for Multiple Finite Populations]\label{Prop:Hoeffdings1}
Let $\mathcal{Y} =\{1,\cdots, n\}$ and suppose $\{X_{j}\}_{j=1}^l $ be random samples drawn without replacement from $\mathcal{Y}$. Consider $M$ distinct finite populations $\{\mathcal{X}^{(m)}\}_{m=1}^M$ with each $\mathcal{X}^{(m)} = \{ x^{(m)}_j\}_{j=1}^n  \subseteq [a, b]$, 
i.e., $\max_j x_j^{(m)} \le b$ and $\min_j x_j^{(m)} \ge a$ for all $m$. Define the population mean of $\mathcal{X}^{(m)}$ as $ \mu^{(m)} := \frac{1}{l}\sum_{j=1}^l x^{(m)}_j$. From each population $\mathcal{X}^{(m)}$, choose a random subset $\{x^{(m)}_{X_i}\}_{i=1}^l$ and let the sample sum be $S^{(m)} := \sum_{i=1}^l x^{(m)}_{X_i}$. Then, for any $t > 0$, 
$$
\Pr \left( \max_{m \in [M] } \left\lvert S^{(m)} - l \mu^{(m)}\right\rvert  \ge l t \right)
\leq 2M e^{-2 l t^2/(b-a)^2}.
$$
\end{proposition}

\begin{proposition}[Hoeffding’s Inequality for Multiple Populations with Batched Sampling]\label{Prop:Hoeffdings2} Let $\mathcal{Y} =\{1,\cdots, n\}$, $n\geq 3Kl$, and suppose $\{\{X_{lk+i}\}_{i=1}^l\}_{k=0,\cdots,K-1} $ is a random sample drawn without replacement from $\mathcal{Y}$. Consider $M$ distinct finite populations $\{\mathcal{X}^{(m)}\}_{m=1}^M$ with each $\mathcal{X}^{(m)} = \{ x^{(m)}_{j}\}_{j=1}^n  \subseteq [a, b]$, i.e., $\max_j x_j^{(m)} \le b$ and $\min_j x_j^{(m)} \ge a$ for all $m$. Define the population mean of $\mathcal{X}^{(m)}$ as $ \mu^{(m)} := \frac{1}{l}\sum_{j=1}^l x^{(m)}_j$. From each population $\mathcal{X}^{(m)}$, obtain random $K$ subsets of indices $\{\{x^{(m)}_{X_{lk+i}}\}_{i=1}^l \}_{k=0,\cdots,K-1}$. Let the sample sum be $S^{(m)} := \sum_{i=1}^l x^{(m)}_{X_{lk+i}}$. Then, for any $t > 0$, 
$$
\Pr \left( \max_{\substack{m \in [M]\\
k\in \{0,\cdots,K-1\} }} \left\lvert S^{(m)} - l \mu^{(m)}\right\rvert  \ge l t \right)
\leq 2KM e^{-\frac{2 l t^2}{4(b-a)^2}}.
$$
\end{proposition}

\subsection{Pseudocode of Algorithms}\label{app:algo}
In this appendix, we provide the pseudocode for the main algorithms that form the basis of our QTERM framework, from Ref. \cite{fanizza_learning_2024}.
\begin{algorithm}[H]
\caption{Quantum Threshold Search on Nonidentical States (\textbf{ThresholdSearch})} 
\label{alg:threshold_search}
\begin{algorithmic}[1]
\Require $\varepsilon, C_1, C_2> 0$.
\Ensure Of the pairs $\left\{\left(\Pi_1^{(c)}, \ldots, \Pi_n^{(c)}, \theta_c\right)\right\}_{c=1}^m$, outputs $c^\star \in [m]$ such that
$$\frac{1}{n} \sum_{i=1}^n \mathrm{Tr}\left[\Pi_i^{(c)} \rho_i\right] > \theta_{c^\star}.
$$

\State \textbf{Initialize:} $\rho^{(0)} \gets \rho_1 \otimes \ldots \otimes \rho_n$.
\For{$c = 1, \ldots, m$}
    \State \textbf{Choose:} A pair $\left(\left\{\Pi_1^{(c)}, \ldots, \Pi_n^{(c)}\right\}, \theta_c\right)$.
    \State \textbf{Measurement:} Apply the two-outcome POVM $\left\{B_c, \bar{B}_c := 1 - B_c\right\}$, where $B_c$ is constructed as in Theorem 8, with threshold $\theta_c - \varepsilon$ and with conditions on $C_1,C_2$, on $\rho^{(c-1)}$. \State Let $\rho^{(c)}$ denote the postmeasurement state.
    \If{the measurement accepts ($B_c$)}
        \State \textbf{Return:} $c$.
    \EndIf
\EndFor
\If{no measurement is accepted (all hypotheses rejected)}
    \State \textbf{Return} 
\EndIf
\end{algorithmic}
\end{algorithm}

\begin{algorithm}[H]
\caption{Learning projector-valued functions}
\label{alg:learning_projector_valued}
\begin{algorithmic}[1]
\Require $2Tk$ product states $\{\rho_s\}_{s\in[2Tk]}$ and $2Tkm$ sets of projectors 
$\{\{h_s^{(c)}\}_{c\in [m]}\}_{s\in[2Tk]}$, 
where $h_s^{(c)} = \{\Pi_{s,1}^{(c)}, \ldots, \Pi_{s,l}^{(c)}\}$. Parameters: $\varepsilon, \delta, k > 0$. 
\State \textbf{Initialize:} $\theta = 1/2$, \texttt{low} $= 0$, \texttt{high} $= 1$, \texttt{failures} $= 0$, $s = 0$.

\While{\texttt{high} $-$ \texttt{low} $\geq 6 \varepsilon$}
    \If{\texttt{failures} $< k$}
        \State $s \gets s + 1$
        \State \textbf{Threshold Search:} Apply Algorithm \ref{alg:threshold_search} on the set of projector-threshold pairs 
        $\left\{\left(h_{2s-1}^{(c)}, \theta - \varepsilon \right)\right\}_{c=1}^m$ 
        with parameter $\varepsilon / 4$ to the product state $\rho_{2s-1}$.
        \If{\texttt{ThresholdSearch} doesn't output a concept}
            \State \texttt{failures} $\gets$ \texttt{failures} $+ 1$
        \Else
            \State \textbf{ThresholdSearch outputs concept $c$:}
            \State \textbf{Check:} Measure $\overline{h_{2s-1}^{(c)}}$ on $\rho_{2s-1}$ and check if $X_{c, 2s-1} \geq n(\theta - 7/4 \varepsilon)$.
            \If{Check outputs 'yes'}
                \State \texttt{low} $\gets \theta - 2 \varepsilon$, \texttt{high} remains unchanged
                \State $\theta \gets \frac{1}{2}($\texttt{high} $+$ \texttt{low}$)$ \Comment{Update interval to upper half}
                \State \texttt{failures} $\gets 0$
            \Else
                \State Check outputs 'no'
                \State \texttt{failures} $\gets$ \texttt{failures} $+ 1$
            \EndIf
        \EndIf
    \Else
        \State $k$ consecutive failures occurred
        \State \texttt{low} remains unchanged, \texttt{high} $\gets \theta$
        \State $\theta \gets \frac{1}{2}($\texttt{high} $+$ \texttt{low}$)$ \Comment{Update interval to lower half}
        \State \texttt{failures}
    \EndIf
\EndWhile

\State \textbf{Output:} $\theta$ and the last selected concept, if there is one; otherwise, pick one randomly.
\end{algorithmic}
\end{algorithm}

\bibliography{bibliography}

\begin{thebibliography}{71}%
\makeatletter
\providecommand \@ifxundefined [1]{%
 \@ifx{#1\undefined}
}%
\providecommand \@ifnum [1]{%
 \ifnum #1\expandafter \@firstoftwo
 \else \expandafter \@secondoftwo
 \fi
}%
\providecommand \@ifx [1]{%
 \ifx #1\expandafter \@firstoftwo
 \else \expandafter \@secondoftwo
 \fi
}%
\providecommand \natexlab [1]{#1}%
\providecommand \enquote  [1]{``#1''}%
\providecommand \bibnamefont  [1]{#1}%
\providecommand \bibfnamefont [1]{#1}%
\providecommand \citenamefont [1]{#1}%
\providecommand \href@noop [0]{\@secondoftwo}%
\providecommand \href [0]{\begingroup \@sanitize@url \@href}%
\providecommand \@href[1]{\@@startlink{#1}\@@href}%
\providecommand \@@href[1]{\endgroup#1\@@endlink}%
\providecommand \@sanitize@url [0]{\catcode `\\12\catcode `\$12\catcode `\&12\catcode `\#12\catcode `\^12\catcode `\_12\catcode `\%12\relax}%
\providecommand \@@startlink[1]{}%
\providecommand \@@endlink[0]{}%
\providecommand \url  [0]{\begingroup\@sanitize@url \@url }%
\providecommand \@url [1]{\endgroup\@href {#1}{\urlprefix }}%
\providecommand \urlprefix  [0]{URL }%
\providecommand \Eprint [0]{\href }%
\providecommand \doibase [0]{https://doi.org/}%
\providecommand \selectlanguage [0]{\@gobble}%
\providecommand \bibinfo  [0]{\@secondoftwo}%
\providecommand \bibfield  [0]{\@secondoftwo}%
\providecommand \translation [1]{[#1]}%
\providecommand \BibitemOpen [0]{}%
\providecommand \bibitemStop [0]{}%
\providecommand \bibitemNoStop [0]{.\EOS\space}%
\providecommand \EOS [0]{\spacefactor3000\relax}%
\providecommand \BibitemShut  [1]{\csname bibitem#1\endcsname}%
\let\auto@bib@innerbib\@empty
\bibitem [{\citenamefont {Valiant}(1984)}]{valiant1984theory}%
  \BibitemOpen
  \bibfield  {author} {\bibinfo {author} {\bibfnamefont {L.}~\bibnamefont {Valiant}},\ }\bibfield  {title} {\bibinfo {title} {A theory of the learnable},\ }\href@noop {} {\bibfield  {journal} {\bibinfo  {journal} {Communications of the ACM}\ }\textbf {\bibinfo {volume} {27}},\ \bibinfo {pages} {1134} (\bibinfo {year} {1984})}\BibitemShut {NoStop}%
\bibitem [{\citenamefont {Kearns}\ and\ \citenamefont {Vazirani}(1994)}]{kearns1994introduction}%
  \BibitemOpen
  \bibfield  {author} {\bibinfo {author} {\bibfnamefont {M.}~\bibnamefont {Kearns}}\ and\ \bibinfo {author} {\bibfnamefont {U.}~\bibnamefont {Vazirani}},\ }\href@noop {} {\emph {\bibinfo {title} {An Introduction to Computational Learning Theory}}}\ (\bibinfo  {publisher} {MIT Press},\ \bibinfo {year} {1994})\BibitemShut {NoStop}%
\bibitem [{\citenamefont {Kearns}\ and\ \citenamefont {Schapire}(1994)}]{kearns_efficient_1994}%
  \BibitemOpen
  \bibfield  {author} {\bibinfo {author} {\bibfnamefont {M.~J.}\ \bibnamefont {Kearns}}\ and\ \bibinfo {author} {\bibfnamefont {R.~E.}\ \bibnamefont {Schapire}},\ }\bibfield  {title} {\bibinfo {title} {Efficient distribution-free learning of probabilistic concepts},\ }\href {https://doi.org/https://doi.org/10.1016/S0022-0000(05)80062-5} {\bibfield  {journal} {\bibinfo  {journal} {Journal of Computer and System Sciences}\ }\textbf {\bibinfo {volume} {48}},\ \bibinfo {pages} {464} (\bibinfo {year} {1994})}\BibitemShut {NoStop}%
\bibitem [{\citenamefont {Bishop}(2006)}]{bishop_pattern_2006}%
  \BibitemOpen
  \bibfield  {author} {\bibinfo {author} {\bibfnamefont {C.~M.}\ \bibnamefont {Bishop}},\ }\href@noop {} {\emph {\bibinfo {title} {Pattern {Recognition} and {Machine} {Learning} ({Information} {Science} and {Statistics})}}}\ (\bibinfo  {publisher} {Springer-Verlag},\ \bibinfo {address} {Berlin, Heidelberg},\ \bibinfo {year} {2006})\BibitemShut {NoStop}%
\bibitem [{\citenamefont {Goodfellow}\ \emph {et~al.}(2016)\citenamefont {Goodfellow}, \citenamefont {Bengio},\ and\ \citenamefont {Courville}}]{goodfellow_deep_2016}%
  \BibitemOpen
  \bibfield  {author} {\bibinfo {author} {\bibfnamefont {I.}~\bibnamefont {Goodfellow}}, \bibinfo {author} {\bibfnamefont {Y.}~\bibnamefont {Bengio}},\ and\ \bibinfo {author} {\bibfnamefont {A.}~\bibnamefont {Courville}},\ }\href@noop {} {\emph {\bibinfo {title} {Deep Learning}}}\ (\bibinfo  {publisher} {MIT Press},\ \bibinfo {year} {2016})\ \bibinfo {note} {\url{http://www.deeplearningbook.org}}\BibitemShut {NoStop}%
\bibitem [{\citenamefont {LeCun}\ \emph {et~al.}(2015)\citenamefont {LeCun}, \citenamefont {Bengio},\ and\ \citenamefont {Hinton}}]{lecun_deep_2015}%
  \BibitemOpen
  \bibfield  {author} {\bibinfo {author} {\bibfnamefont {Y.}~\bibnamefont {LeCun}}, \bibinfo {author} {\bibfnamefont {Y.}~\bibnamefont {Bengio}},\ and\ \bibinfo {author} {\bibfnamefont {G.}~\bibnamefont {Hinton}},\ }\bibfield  {title} {\bibinfo {title} {Deep learning},\ }\href@noop {} {\bibfield  {journal} {\bibinfo  {journal} {Nature}\ }\textbf {\bibinfo {volume} {521}},\ \bibinfo {pages} {436} (\bibinfo {year} {2015})}\BibitemShut {NoStop}%
\bibitem [{\citenamefont {Li}\ \emph {et~al.}(2020)\citenamefont {Li}, \citenamefont {Beirami}, \citenamefont {Sanjabi},\ and\ \citenamefont {Smith}}]{li2020tilted}%
  \BibitemOpen
  \bibfield  {author} {\bibinfo {author} {\bibfnamefont {T.}~\bibnamefont {Li}}, \bibinfo {author} {\bibfnamefont {A.}~\bibnamefont {Beirami}}, \bibinfo {author} {\bibfnamefont {M.}~\bibnamefont {Sanjabi}},\ and\ \bibinfo {author} {\bibfnamefont {V.}~\bibnamefont {Smith}},\ }\bibfield  {title} {\bibinfo {title} {Tilted empirical risk minimization},\ }\href@noop {} {\bibfield  {journal} {\bibinfo  {journal} {arXiv preprint arXiv:2007.01162}\ } (\bibinfo {year} {2020})}\BibitemShut {NoStop}%
\bibitem [{\citenamefont {Li}\ \emph {et~al.}(2023)\citenamefont {Li}, \citenamefont {Beirami}, \citenamefont {Sanjabi},\ and\ \citenamefont {Smith}}]{li2023tilted}%
  \BibitemOpen
  \bibfield  {author} {\bibinfo {author} {\bibfnamefont {T.}~\bibnamefont {Li}}, \bibinfo {author} {\bibfnamefont {A.}~\bibnamefont {Beirami}}, \bibinfo {author} {\bibfnamefont {M.}~\bibnamefont {Sanjabi}},\ and\ \bibinfo {author} {\bibfnamefont {V.}~\bibnamefont {Smith}},\ }\bibfield  {title} {\bibinfo {title} {On tilted losses in machine learning: Theory and applications},\ }\href@noop {} {\bibfield  {journal} {\bibinfo  {journal} {Journal of Machine Learning Research}\ }\textbf {\bibinfo {volume} {24}},\ \bibinfo {pages} {1} (\bibinfo {year} {2023})}\BibitemShut {NoStop}%
\bibitem [{\citenamefont {Siegmund}(1976)}]{siegmund1976importance}%
  \BibitemOpen
  \bibfield  {author} {\bibinfo {author} {\bibfnamefont {D.}~\bibnamefont {Siegmund}},\ }\bibfield  {title} {\bibinfo {title} {Importance sampling in the monte carlo study of sequential tests},\ }\href@noop {} {\bibfield  {journal} {\bibinfo  {journal} {The Annals of Statistics}\ ,\ \bibinfo {pages} {673}} (\bibinfo {year} {1976})}\BibitemShut {NoStop}%
\bibitem [{\citenamefont {Butler}(2007)}]{butler2007saddlepoint}%
  \BibitemOpen
  \bibfield  {author} {\bibinfo {author} {\bibfnamefont {R.~W.}\ \bibnamefont {Butler}},\ }\href@noop {} {\emph {\bibinfo {title} {Saddlepoint approximations with applications}}},\ Vol.~\bibinfo {volume} {22}\ (\bibinfo  {publisher} {Cambridge University Press},\ \bibinfo {year} {2007})\BibitemShut {NoStop}%
\bibitem [{\citenamefont {Thomas}\ and\ \citenamefont {Joy}(2006)}]{thomas2006elements}%
  \BibitemOpen
  \bibfield  {author} {\bibinfo {author} {\bibfnamefont {M.}~\bibnamefont {Thomas}}\ and\ \bibinfo {author} {\bibfnamefont {A.~T.}\ \bibnamefont {Joy}},\ }\href@noop {} {\emph {\bibinfo {title} {Elements of information theory}}}\ (\bibinfo  {publisher} {Wiley-Interscience},\ \bibinfo {year} {2006})\BibitemShut {NoStop}%
\bibitem [{\citenamefont {Dembo}(2009)}]{dembo2009large}%
  \BibitemOpen
  \bibfield  {author} {\bibinfo {author} {\bibfnamefont {A.}~\bibnamefont {Dembo}},\ }\href@noop {} {\emph {\bibinfo {title} {Large deviations techniques and applications}}}\ (\bibinfo  {publisher} {Springer},\ \bibinfo {year} {2009})\BibitemShut {NoStop}%
\bibitem [{\citenamefont {Aminian}\ \emph {et~al.}(2024)\citenamefont {Aminian}, \citenamefont {Asadi}, \citenamefont {Li}, \citenamefont {Beirami}, \citenamefont {Reinert},\ and\ \citenamefont {Cohen}}]{aminian2024generalizationerrortiltedempirical}%
  \BibitemOpen
  \bibfield  {author} {\bibinfo {author} {\bibfnamefont {G.}~\bibnamefont {Aminian}}, \bibinfo {author} {\bibfnamefont {A.~R.}\ \bibnamefont {Asadi}}, \bibinfo {author} {\bibfnamefont {T.}~\bibnamefont {Li}}, \bibinfo {author} {\bibfnamefont {A.}~\bibnamefont {Beirami}}, \bibinfo {author} {\bibfnamefont {G.}~\bibnamefont {Reinert}},\ and\ \bibinfo {author} {\bibfnamefont {S.~N.}\ \bibnamefont {Cohen}},\ }\href {https://arxiv.org/abs/2409.19431} {\bibinfo {title} {Generalization error of the tilted empirical risk}} (\bibinfo {year} {2024}),\ \Eprint {https://arxiv.org/abs/2409.19431} {arXiv:2409.19431 [stat.ML]} \BibitemShut {NoStop}%
\bibitem [{\citenamefont {Smith}\ \emph {et~al.}(2021)\citenamefont {Smith}, \citenamefont {Dherin}, \citenamefont {Barrett},\ and\ \citenamefont {De}}]{smith2021origin}%
  \BibitemOpen
  \bibfield  {author} {\bibinfo {author} {\bibfnamefont {S.~L.}\ \bibnamefont {Smith}}, \bibinfo {author} {\bibfnamefont {B.}~\bibnamefont {Dherin}}, \bibinfo {author} {\bibfnamefont {D.~G.}\ \bibnamefont {Barrett}},\ and\ \bibinfo {author} {\bibfnamefont {S.}~\bibnamefont {De}},\ }\bibfield  {title} {\bibinfo {title} {On the origin of implicit regularization in stochastic gradient descent},\ }\href@noop {} {\bibfield  {journal} {\bibinfo  {journal} {arXiv preprint arXiv:2101.12176}\ } (\bibinfo {year} {2021})}\BibitemShut {NoStop}%
\bibitem [{\citenamefont {Bauer}\ \emph {et~al.}(2007)\citenamefont {Bauer}, \citenamefont {Pereverzev},\ and\ \citenamefont {Rosasco}}]{bauer_regularization_2007}%
  \BibitemOpen
  \bibfield  {author} {\bibinfo {author} {\bibfnamefont {F.}~\bibnamefont {Bauer}}, \bibinfo {author} {\bibfnamefont {S.}~\bibnamefont {Pereverzev}},\ and\ \bibinfo {author} {\bibfnamefont {L.}~\bibnamefont {Rosasco}},\ }\bibfield  {title} {\bibinfo {title} {On regularization algorithms in learning theory},\ }\href {https://doi.org/https://doi.org/10.1016/j.jco.2006.07.001} {\bibfield  {journal} {\bibinfo  {journal} {Journal of Complexity}\ }\textbf {\bibinfo {volume} {23}},\ \bibinfo {pages} {52} (\bibinfo {year} {2007})}\BibitemShut {NoStop}%
\bibitem [{\citenamefont {Mol}\ \emph {et~al.}(2008)\citenamefont {Mol}, \citenamefont {Vito},\ and\ \citenamefont {Rosasco}}]{demol2008elasticnetregularizationlearningtheory}%
  \BibitemOpen
  \bibfield  {author} {\bibinfo {author} {\bibfnamefont {C.~D.}\ \bibnamefont {Mol}}, \bibinfo {author} {\bibfnamefont {E.~D.}\ \bibnamefont {Vito}},\ and\ \bibinfo {author} {\bibfnamefont {L.}~\bibnamefont {Rosasco}},\ }\href {https://arxiv.org/abs/0807.3423} {\bibinfo {title} {Elastic-net regularization in learning theory}} (\bibinfo {year} {2008}),\ \Eprint {https://arxiv.org/abs/0807.3423} {arXiv:0807.3423 [stat.ML]} \BibitemShut {NoStop}%
\bibitem [{\citenamefont {Schuld}\ and\ \citenamefont {Petruccione}(2021)}]{schuld_machine_2021}%
  \BibitemOpen
  \bibfield  {author} {\bibinfo {author} {\bibfnamefont {M.}~\bibnamefont {Schuld}}\ and\ \bibinfo {author} {\bibfnamefont {F.}~\bibnamefont {Petruccione}},\ }\href {https://books.google.com.au/books?id=-N5IEAAAQBAJ} {\emph {\bibinfo {title} {Machine Learning with Quantum Computers}}},\ Quantum Science and Technology\ (\bibinfo  {publisher} {Springer International Publishing},\ \bibinfo {year} {2021})\BibitemShut {NoStop}%
\bibitem [{\citenamefont {Biamonte}\ \emph {et~al.}(2016)\citenamefont {Biamonte}, \citenamefont {Wittek}, \citenamefont {Pancotti}, \citenamefont {Rebentrost}, \citenamefont {Wiebe},\ and\ \citenamefont {Lloyd}}]{biamonte_quantum_2017}%
  \BibitemOpen
  \bibfield  {author} {\bibinfo {author} {\bibfnamefont {J.}~\bibnamefont {Biamonte}}, \bibinfo {author} {\bibfnamefont {P.}~\bibnamefont {Wittek}}, \bibinfo {author} {\bibfnamefont {N.}~\bibnamefont {Pancotti}}, \bibinfo {author} {\bibfnamefont {P.}~\bibnamefont {Rebentrost}}, \bibinfo {author} {\bibfnamefont {N.}~\bibnamefont {Wiebe}},\ and\ \bibinfo {author} {\bibfnamefont {S.}~\bibnamefont {Lloyd}},\ }\bibfield  {title} {\bibinfo {title} {Quantum {Machine} {Learning}},\ }\href {https://doi.org/10.1038/nature23474} {\bibfield  {journal} {\bibinfo  {journal} {Nature}\ }\textbf {\bibinfo {volume} {549}},\ \bibinfo {pages} {195} (\bibinfo {year} {2016})}\BibitemShut {NoStop}%
\bibitem [{\citenamefont {Rebentrost}\ \emph {et~al.}(2014)\citenamefont {Rebentrost}, \citenamefont {Mohseni},\ and\ \citenamefont {Lloyd}}]{rebentrost_quantum_2014}%
  \BibitemOpen
  \bibfield  {author} {\bibinfo {author} {\bibfnamefont {P.}~\bibnamefont {Rebentrost}}, \bibinfo {author} {\bibfnamefont {M.}~\bibnamefont {Mohseni}},\ and\ \bibinfo {author} {\bibfnamefont {S.}~\bibnamefont {Lloyd}},\ }\bibfield  {title} {\bibinfo {title} {Quantum support vector machine for big data classification},\ }\href@noop {} {\bibfield  {journal} {\bibinfo  {journal} {Physical review letters}\ }\textbf {\bibinfo {volume} {113}},\ \bibinfo {pages} {130503} (\bibinfo {year} {2014})}\BibitemShut {NoStop}%
\bibitem [{\citenamefont {Schuld}\ \emph {et~al.}(2020)\citenamefont {Schuld}, \citenamefont {Bocharov}, \citenamefont {Svore},\ and\ \citenamefont {Wiebe}}]{schuld2020circuit}%
  \BibitemOpen
  \bibfield  {author} {\bibinfo {author} {\bibfnamefont {M.}~\bibnamefont {Schuld}}, \bibinfo {author} {\bibfnamefont {A.}~\bibnamefont {Bocharov}}, \bibinfo {author} {\bibfnamefont {K.~M.}\ \bibnamefont {Svore}},\ and\ \bibinfo {author} {\bibfnamefont {N.}~\bibnamefont {Wiebe}},\ }\bibfield  {title} {\bibinfo {title} {Circuit-centric quantum classifiers},\ }\href {https://doi.org/10.1103/PhysRevA.101.032308} {\bibfield  {journal} {\bibinfo  {journal} {Phys. Rev. A}\ }\textbf {\bibinfo {volume} {101}},\ \bibinfo {pages} {032308} (\bibinfo {year} {2020})}\BibitemShut {NoStop}%
\bibitem [{\citenamefont {Beer}\ \emph {et~al.}(2020)\citenamefont {Beer}, \citenamefont {Bondarenko}, \citenamefont {Farrelly}, \citenamefont {Osborne}, \citenamefont {Salzmann}, \citenamefont {Scheiermann},\ and\ \citenamefont {Wolf}}]{beer_training_2020}%
  \BibitemOpen
  \bibfield  {author} {\bibinfo {author} {\bibfnamefont {K.}~\bibnamefont {Beer}}, \bibinfo {author} {\bibfnamefont {D.}~\bibnamefont {Bondarenko}}, \bibinfo {author} {\bibfnamefont {T.}~\bibnamefont {Farrelly}}, \bibinfo {author} {\bibfnamefont {T.~J.}\ \bibnamefont {Osborne}}, \bibinfo {author} {\bibfnamefont {R.}~\bibnamefont {Salzmann}}, \bibinfo {author} {\bibfnamefont {D.}~\bibnamefont {Scheiermann}},\ and\ \bibinfo {author} {\bibfnamefont {R.}~\bibnamefont {Wolf}},\ }\bibfield  {title} {\bibinfo {title} {Training deep quantum neural networks},\ }\href@noop {} {\bibfield  {journal} {\bibinfo  {journal} {Nature communications}\ }\textbf {\bibinfo {volume} {11}},\ \bibinfo {pages} {808} (\bibinfo {year} {2020})}\BibitemShut {NoStop}%
\bibitem [{\citenamefont {Nielsen}\ and\ \citenamefont {Chuang}(2011)}]{nielsen_quantum_2011}%
  \BibitemOpen
  \bibfield  {author} {\bibinfo {author} {\bibfnamefont {M.~A.}\ \bibnamefont {Nielsen}}\ and\ \bibinfo {author} {\bibfnamefont {I.~L.}\ \bibnamefont {Chuang}},\ }\href@noop {} {\emph {\bibinfo {title} {Quantum Computation and Quantum Information: 10th Anniversary Edition}}}\ (\bibinfo  {publisher} {Cambridge University Press},\ \bibinfo {year} {2011})\BibitemShut {NoStop}%
\bibitem [{\citenamefont {Montanaro}(2016)}]{montanaro_quantum_2016}%
  \BibitemOpen
  \bibfield  {author} {\bibinfo {author} {\bibfnamefont {A.}~\bibnamefont {Montanaro}},\ }\bibfield  {title} {\bibinfo {title} {Quantum algorithms: an overview},\ }\bibfield  {journal} {\bibinfo  {journal} {npj Quantum Information}\ }\textbf {\bibinfo {volume} {2}},\ \href {https://doi.org/10.1038/npjqi.2015.23} {10.1038/npjqi.2015.23} (\bibinfo {year} {2016})\BibitemShut {NoStop}%
\bibitem [{\citenamefont {Preskill}(2018)}]{preskill_quantum_2018}%
  \BibitemOpen
  \bibfield  {author} {\bibinfo {author} {\bibfnamefont {J.}~\bibnamefont {Preskill}},\ }\bibfield  {title} {\bibinfo {title} {Quantum computing in the nisq era and beyond},\ }\href {https://doi.org/10.22331/q-2018-08-06-79} {\bibfield  {journal} {\bibinfo  {journal} {Quantum}\ }\textbf {\bibinfo {volume} {2}},\ \bibinfo {pages} {79} (\bibinfo {year} {2018})}\BibitemShut {NoStop}%
\bibitem [{\citenamefont {Shor}(1996)}]{shor1996fault}%
  \BibitemOpen
  \bibfield  {author} {\bibinfo {author} {\bibfnamefont {P.~W.}\ \bibnamefont {Shor}},\ }\bibfield  {title} {\bibinfo {title} {Fault-tolerant quantum computation},\ }in\ \href@noop {} {\emph {\bibinfo {booktitle} {Proceedings of 37th conference on foundations of computer science}}}\ (\bibinfo {organization} {IEEE},\ \bibinfo {year} {1996})\ pp.\ \bibinfo {pages} {56--65}\BibitemShut {NoStop}%
\bibitem [{\citenamefont {Cerezo}\ \emph {et~al.}(2021)\citenamefont {Cerezo}, \citenamefont {Arrasmith}, \citenamefont {Babbush}, \citenamefont {Benjamin}, \citenamefont {Endo}, \citenamefont {Fujii}, \citenamefont {McClean}, \citenamefont {Mitarai}, \citenamefont {Yuan}, \citenamefont {Cincio},\ and\ \citenamefont {Coles}}]{cerezo2021variational}%
  \BibitemOpen
  \bibfield  {author} {\bibinfo {author} {\bibfnamefont {M.}~\bibnamefont {Cerezo}}, \bibinfo {author} {\bibfnamefont {A.}~\bibnamefont {Arrasmith}}, \bibinfo {author} {\bibfnamefont {R.}~\bibnamefont {Babbush}}, \bibinfo {author} {\bibfnamefont {S.~C.}\ \bibnamefont {Benjamin}}, \bibinfo {author} {\bibfnamefont {S.}~\bibnamefont {Endo}}, \bibinfo {author} {\bibfnamefont {K.}~\bibnamefont {Fujii}}, \bibinfo {author} {\bibfnamefont {J.~R.}\ \bibnamefont {McClean}}, \bibinfo {author} {\bibfnamefont {K.}~\bibnamefont {Mitarai}}, \bibinfo {author} {\bibfnamefont {X.}~\bibnamefont {Yuan}}, \bibinfo {author} {\bibfnamefont {L.}~\bibnamefont {Cincio}},\ and\ \bibinfo {author} {\bibfnamefont {P.~J.}\ \bibnamefont {Coles}},\ }\bibfield  {title} {\bibinfo {title} {Variational quantum algorithms},\ }\href {http://doi.org/10.1038/s42254-021-00348-9} {\bibfield  {journal} {\bibinfo  {journal} {Nat. Rev. Phys}\ }\textbf {\bibinfo {volume} {3}},\ \bibinfo {pages} {625–644} (\bibinfo {year} {2021})}\BibitemShut {NoStop}%
\bibitem [{\citenamefont {Benedetti}\ \emph {et~al.}(2019)\citenamefont {Benedetti}, \citenamefont {Lloyd}, \citenamefont {Sack},\ and\ \citenamefont {Fiorentini}}]{benedetti2019parameterized}%
  \BibitemOpen
  \bibfield  {author} {\bibinfo {author} {\bibfnamefont {M.}~\bibnamefont {Benedetti}}, \bibinfo {author} {\bibfnamefont {E.}~\bibnamefont {Lloyd}}, \bibinfo {author} {\bibfnamefont {S.}~\bibnamefont {Sack}},\ and\ \bibinfo {author} {\bibfnamefont {M.}~\bibnamefont {Fiorentini}},\ }\bibfield  {title} {\bibinfo {title} {Parameterized quantum circuits as machine learning models},\ }\href {https://doi.org/10.1088/2058-9565/ab4eb5} {\bibfield  {journal} {\bibinfo  {journal} {Quantum Science and Technology}\ }\textbf {\bibinfo {volume} {4}},\ \bibinfo {pages} {043001} (\bibinfo {year} {2019})}\BibitemShut {NoStop}%
\bibitem [{\citenamefont {Gily\'{e}n}\ \emph {et~al.}(2019)\citenamefont {Gily\'{e}n}, \citenamefont {Su}, \citenamefont {Low},\ and\ \citenamefont {Wiebe}}]{gilyen_quantum_2019}%
  \BibitemOpen
  \bibfield  {author} {\bibinfo {author} {\bibfnamefont {A.}~\bibnamefont {Gily\'{e}n}}, \bibinfo {author} {\bibfnamefont {Y.}~\bibnamefont {Su}}, \bibinfo {author} {\bibfnamefont {G.~H.}\ \bibnamefont {Low}},\ and\ \bibinfo {author} {\bibfnamefont {N.}~\bibnamefont {Wiebe}},\ }\bibfield  {title} {\bibinfo {title} {Quantum singular value transformation and beyond: Exponential improvements for quantum matrix arithmetics},\ }in\ \href {https://doi.org/10.1145/3313276.3316366} {\emph {\bibinfo {booktitle} {Proceedings of the 51st Annual ACM SIGACT Symposium on Theory of Computing}}},\ \bibinfo {series and number} {STOC 2019}\ (\bibinfo  {publisher} {Association for Computing Machinery},\ \bibinfo {address} {New York, NY, USA},\ \bibinfo {year} {2019})\ p.\ \bibinfo {pages} {193–204}\BibitemShut {NoStop}%
\bibitem [{\citenamefont {Ivashkov}\ \emph {et~al.}(2024)\citenamefont {Ivashkov}, \citenamefont {Huang}, \citenamefont {Koor}, \citenamefont {Pira},\ and\ \citenamefont {Rebentrost}}]{ivashkov2024qkanquantumkolmogorovarnoldnetworks}%
  \BibitemOpen
  \bibfield  {author} {\bibinfo {author} {\bibfnamefont {P.}~\bibnamefont {Ivashkov}}, \bibinfo {author} {\bibfnamefont {P.-W.}\ \bibnamefont {Huang}}, \bibinfo {author} {\bibfnamefont {K.}~\bibnamefont {Koor}}, \bibinfo {author} {\bibfnamefont {L.}~\bibnamefont {Pira}},\ and\ \bibinfo {author} {\bibfnamefont {P.}~\bibnamefont {Rebentrost}},\ }\href {https://arxiv.org/abs/2410.04435} {\bibinfo {title} {Qkan: Quantum kolmogorov-arnold networks}} (\bibinfo {year} {2024}),\ \Eprint {https://arxiv.org/abs/2410.04435} {arXiv:2410.04435 [quant-ph]} \BibitemShut {NoStop}%
\bibitem [{\citenamefont {Guo}\ \emph {et~al.}(2024)\citenamefont {Guo}, \citenamefont {Yu}, \citenamefont {Agrawal},\ and\ \citenamefont {Rebentrost}}]{guo_quantum_2024}%
  \BibitemOpen
  \bibfield  {author} {\bibinfo {author} {\bibfnamefont {N.}~\bibnamefont {Guo}}, \bibinfo {author} {\bibfnamefont {Z.}~\bibnamefont {Yu}}, \bibinfo {author} {\bibfnamefont {A.}~\bibnamefont {Agrawal}},\ and\ \bibinfo {author} {\bibfnamefont {P.}~\bibnamefont {Rebentrost}},\ }\href@noop {} {\bibinfo {title} {Quantum linear algebra is all you need for transformer architectures}} (\bibinfo {year} {2024}),\ \Eprint {https://arxiv.org/abs/2402.16714} {arXiv:2402.16714 [quant-ph]} \BibitemShut {NoStop}%
\bibitem [{\citenamefont {Fanizza}\ \emph {et~al.}(2024)\citenamefont {Fanizza}, \citenamefont {Quek},\ and\ \citenamefont {Rosati}}]{fanizza_learning_2024}%
  \BibitemOpen
  \bibfield  {author} {\bibinfo {author} {\bibfnamefont {M.}~\bibnamefont {Fanizza}}, \bibinfo {author} {\bibfnamefont {Y.}~\bibnamefont {Quek}},\ and\ \bibinfo {author} {\bibfnamefont {M.}~\bibnamefont {Rosati}},\ }\bibfield  {title} {\bibinfo {title} {Learning quantum processes without input control},\ }\href {https://doi.org/10.1103/PRXQuantum.5.020367} {\bibfield  {journal} {\bibinfo  {journal} {PRX Quantum}\ }\textbf {\bibinfo {volume} {5}},\ \bibinfo {pages} {020367} (\bibinfo {year} {2024})}\BibitemShut {NoStop}%
\bibitem [{\citenamefont {Vapnik}\ and\ \citenamefont {Chervonenkis}(1974)}]{vapnik1974teoriya}%
  \BibitemOpen
  \bibfield  {author} {\bibinfo {author} {\bibfnamefont {V.~N.}\ \bibnamefont {Vapnik}}\ and\ \bibinfo {author} {\bibfnamefont {A.~Y.}\ \bibnamefont {Chervonenkis}},\ }\href@noop {} {{\selectlanguage {Russian}\emph {\bibinfo {title} {Teoriya raspoznavaniya obrazov [Theory of Pattern Recognition]}}}}\ (\bibinfo  {publisher} {Nauka},\ \bibinfo {address} {Moscow},\ \bibinfo {year} {1974})\BibitemShut {NoStop}%
\bibitem [{\citenamefont {Gyurik}\ \emph {et~al.}(2023)\citenamefont {Gyurik}, \citenamefont {Vreumingen},\ and\ \citenamefont {Dunjko}}]{gyurik_structuralrisk_2023}%
  \BibitemOpen
  \bibfield  {author} {\bibinfo {author} {\bibfnamefont {C.}~\bibnamefont {Gyurik}}, \bibinfo {author} {\bibfnamefont {D.}~\bibnamefont {Vreumingen}, \bibfnamefont {van}},\ and\ \bibinfo {author} {\bibfnamefont {V.}~\bibnamefont {Dunjko}},\ }\bibfield  {title} {\bibinfo {title} {Structural risk minimization for quantum linear classifiers},\ }\href {https://doi.org/10.22331/q-2023-01-13-893} {\bibfield  {journal} {\bibinfo  {journal} {{Quantum}}\ }\textbf {\bibinfo {volume} {7}},\ \bibinfo {pages} {893} (\bibinfo {year} {2023})}\BibitemShut {NoStop}%
\bibitem [{\citenamefont {Heidari}\ \emph {et~al.}(2021)\citenamefont {Heidari}, \citenamefont {Padakandla},\ and\ \citenamefont {Szpankowski}}]{heidari_theoretical_2021}%
  \BibitemOpen
  \bibfield  {author} {\bibinfo {author} {\bibfnamefont {M.}~\bibnamefont {Heidari}}, \bibinfo {author} {\bibfnamefont {A.}~\bibnamefont {Padakandla}},\ and\ \bibinfo {author} {\bibfnamefont {W.}~\bibnamefont {Szpankowski}},\ }\href@noop {} {\bibinfo {title} {A theoretical framework for learning from quantum data}} (\bibinfo {year} {2021}),\ \Eprint {https://arxiv.org/abs/2107.06406} {arXiv:2107.06406 [quant-ph]} \BibitemShut {NoStop}%
\bibitem [{\citenamefont {Padakandla}\ and\ \citenamefont {Magner}(2022)}]{padakandla2022pac}%
  \BibitemOpen
  \bibfield  {author} {\bibinfo {author} {\bibfnamefont {A.}~\bibnamefont {Padakandla}}\ and\ \bibinfo {author} {\bibfnamefont {A.}~\bibnamefont {Magner}},\ }\bibfield  {title} {\bibinfo {title} {Pac learning of quantum measurement classes: Sample complexity bounds and universal consistency},\ }in\ \href@noop {} {\emph {\bibinfo {booktitle} {International Conference on Artificial Intelligence and Statistics}}}\ (\bibinfo {organization} {PMLR},\ \bibinfo {year} {2022})\ pp.\ \bibinfo {pages} {11305--11319}\BibitemShut {NoStop}%
\bibitem [{\citenamefont {Heidari}\ and\ \citenamefont {Szpankowski}(2024)}]{heidari2024new}%
  \BibitemOpen
  \bibfield  {author} {\bibinfo {author} {\bibfnamefont {M.}~\bibnamefont {Heidari}}\ and\ \bibinfo {author} {\bibfnamefont {W.}~\bibnamefont {Szpankowski}},\ }\bibfield  {title} {\bibinfo {title} {New bounds on quantum sample complexity of measurement classes},\ }in\ \href@noop {} {\emph {\bibinfo {booktitle} {2024 IEEE International Symposium on Information Theory (ISIT)}}}\ (\bibinfo {organization} {IEEE},\ \bibinfo {year} {2024})\ pp.\ \bibinfo {pages} {1515--1520}\BibitemShut {NoStop}%
\bibitem [{\citenamefont {Ciliberto}\ \emph {et~al.}(2020)\citenamefont {Ciliberto}, \citenamefont {Rocchetto}, \citenamefont {Rudi},\ and\ \citenamefont {Wossnig}}]{Ciliberto_2020}%
  \BibitemOpen
  \bibfield  {author} {\bibinfo {author} {\bibfnamefont {C.}~\bibnamefont {Ciliberto}}, \bibinfo {author} {\bibfnamefont {A.}~\bibnamefont {Rocchetto}}, \bibinfo {author} {\bibfnamefont {A.}~\bibnamefont {Rudi}},\ and\ \bibinfo {author} {\bibfnamefont {L.}~\bibnamefont {Wossnig}},\ }\bibfield  {title} {\bibinfo {title} {Statistical limits of supervised quantum learning},\ }\bibfield  {journal} {\bibinfo  {journal} {Physical Review A}\ }\textbf {\bibinfo {volume} {102}},\ \href {https://doi.org/10.1103/physreva.102.042414} {10.1103/physreva.102.042414} (\bibinfo {year} {2020})\BibitemShut {NoStop}%
\bibitem [{\citenamefont {Arunachalam}\ and\ \citenamefont {De~Wolf}(2017)}]{arunachalam2017guest}%
  \BibitemOpen
  \bibfield  {author} {\bibinfo {author} {\bibfnamefont {S.}~\bibnamefont {Arunachalam}}\ and\ \bibinfo {author} {\bibfnamefont {R.}~\bibnamefont {De~Wolf}},\ }\bibfield  {title} {\bibinfo {title} {Guest column: A survey of quantum learning theory},\ }\href@noop {} {\bibfield  {journal} {\bibinfo  {journal} {ACM Sigact News}\ }\textbf {\bibinfo {volume} {48}},\ \bibinfo {pages} {41} (\bibinfo {year} {2017})}\BibitemShut {NoStop}%
\bibitem [{\citenamefont {Salmon}\ \emph {et~al.}(2023)\citenamefont {Salmon}, \citenamefont {Strelchuk},\ and\ \citenamefont {Gur}}]{salmon2023provableadvantagequantumpac}%
  \BibitemOpen
  \bibfield  {author} {\bibinfo {author} {\bibfnamefont {W.}~\bibnamefont {Salmon}}, \bibinfo {author} {\bibfnamefont {S.}~\bibnamefont {Strelchuk}},\ and\ \bibinfo {author} {\bibfnamefont {T.}~\bibnamefont {Gur}},\ }\href {https://arxiv.org/abs/2309.10887} {\bibinfo {title} {Provable advantage in quantum pac learning}} (\bibinfo {year} {2023}),\ \Eprint {https://arxiv.org/abs/2309.10887} {arXiv:2309.10887 [quant-ph]} \BibitemShut {NoStop}%
\bibitem [{\citenamefont {Nayak}\ and\ \citenamefont {Sinha}(2024)}]{nayak2024propervsimproperquantum}%
  \BibitemOpen
  \bibfield  {author} {\bibinfo {author} {\bibfnamefont {A.}~\bibnamefont {Nayak}}\ and\ \bibinfo {author} {\bibfnamefont {P.}~\bibnamefont {Sinha}},\ }\href {https://arxiv.org/abs/2403.03295} {\bibinfo {title} {Proper vs improper quantum pac learning}} (\bibinfo {year} {2024}),\ \Eprint {https://arxiv.org/abs/2403.03295} {arXiv:2403.03295 [quant-ph]} \BibitemShut {NoStop}%
\bibitem [{\citenamefont {Chung}\ and\ \citenamefont {Lin}(2021)}]{chung_sample_2021}%
  \BibitemOpen
  \bibfield  {author} {\bibinfo {author} {\bibfnamefont {K.-M.}\ \bibnamefont {Chung}}\ and\ \bibinfo {author} {\bibfnamefont {H.-H.}\ \bibnamefont {Lin}},\ }\bibfield  {title} {\bibinfo {title} {{Sample Efficient Algorithms for Learning Quantum Channels in PAC Model and the Approximate State Discrimination Problem}},\ }in\ \href {https://doi.org/10.4230/LIPIcs.TQC.2021.3} {\emph {\bibinfo {booktitle} {16th Conference on the Theory of Quantum Computation, Communication and Cryptography (TQC 2021)}}},\ \bibinfo {series} {Leibniz International Proceedings in Informatics (LIPIcs)}, Vol.\ \bibinfo {volume} {197},\ \bibinfo {editor} {edited by\ \bibinfo {editor} {\bibfnamefont {M.-H.}\ \bibnamefont {Hsieh}}}\ (\bibinfo  {publisher} {Schloss Dagstuhl -- Leibniz-Zentrum f{\"u}r Informatik},\ \bibinfo {address} {Dagstuhl, Germany},\ \bibinfo {year} {2021})\ pp.\ \bibinfo {pages} {3:1--3:22}\BibitemShut {NoStop}%
\bibitem [{\citenamefont {Caro}\ \emph {et~al.}(2022)\citenamefont {Caro}, \citenamefont {Huang}, \citenamefont {Cerezo}, \citenamefont {Sharma}, \citenamefont {Sornborger}, \citenamefont {Cincio},\ and\ \citenamefont {Coles}}]{caro2022generalization}%
  \BibitemOpen
  \bibfield  {author} {\bibinfo {author} {\bibfnamefont {M.~C.}\ \bibnamefont {Caro}}, \bibinfo {author} {\bibfnamefont {H.-Y.}\ \bibnamefont {Huang}}, \bibinfo {author} {\bibfnamefont {M.}~\bibnamefont {Cerezo}}, \bibinfo {author} {\bibfnamefont {K.}~\bibnamefont {Sharma}}, \bibinfo {author} {\bibfnamefont {A.}~\bibnamefont {Sornborger}}, \bibinfo {author} {\bibfnamefont {L.}~\bibnamefont {Cincio}},\ and\ \bibinfo {author} {\bibfnamefont {P.~J.}\ \bibnamefont {Coles}},\ }\bibfield  {title} {\bibinfo {title} {Generalization in quantum machine learning from few training data},\ }\href {https://doi.org/10.1038/s41467-022-32669-6} {\bibfield  {journal} {\bibinfo  {journal} {Nature Communications}\ }\textbf {\bibinfo {volume} {13}},\ \bibinfo {pages} {4919} (\bibinfo {year} {2022})}\BibitemShut {NoStop}%
\bibitem [{\citenamefont {Abbas}\ \emph {et~al.}(2021)\citenamefont {Abbas}, \citenamefont {Sutter}, \citenamefont {Figalli},\ and\ \citenamefont {Woerner}}]{abbas2021effectivedimensionmachinelearning}%
  \BibitemOpen
  \bibfield  {author} {\bibinfo {author} {\bibfnamefont {A.}~\bibnamefont {Abbas}}, \bibinfo {author} {\bibfnamefont {D.}~\bibnamefont {Sutter}}, \bibinfo {author} {\bibfnamefont {A.}~\bibnamefont {Figalli}},\ and\ \bibinfo {author} {\bibfnamefont {S.}~\bibnamefont {Woerner}},\ }\href {https://arxiv.org/abs/2112.04807} {\bibinfo {title} {Effective dimension of machine learning models}} (\bibinfo {year} {2021}),\ \Eprint {https://arxiv.org/abs/2112.04807} {arXiv:2112.04807 [cs.LG]} \BibitemShut {NoStop}%
\bibitem [{\citenamefont {Gil-Fuster}\ \emph {et~al.}(2024)\citenamefont {Gil-Fuster}, \citenamefont {Eisert},\ and\ \citenamefont {Bravo-Prieto}}]{GilFuster2024}%
  \BibitemOpen
  \bibfield  {author} {\bibinfo {author} {\bibfnamefont {E.}~\bibnamefont {Gil-Fuster}}, \bibinfo {author} {\bibfnamefont {J.}~\bibnamefont {Eisert}},\ and\ \bibinfo {author} {\bibfnamefont {C.}~\bibnamefont {Bravo-Prieto}},\ }\bibfield  {title} {\bibinfo {title} {Understanding quantum machine learning also requires rethinking generalization},\ }\href@noop {} {\bibfield  {journal} {\bibinfo  {journal} {Nature Communications}\ }\textbf {\bibinfo {volume} {15}} (\bibinfo {year} {2024})}\BibitemShut {NoStop}%
\bibitem [{\citenamefont {Bshouty}\ and\ \citenamefont {Jackson}(1999)}]{Bshouty1999}%
  \BibitemOpen
  \bibfield  {author} {\bibinfo {author} {\bibfnamefont {N.~H.}\ \bibnamefont {Bshouty}}\ and\ \bibinfo {author} {\bibfnamefont {J.~C.}\ \bibnamefont {Jackson}},\ }\bibfield  {title} {\bibinfo {title} {Learning dnf over the uniform distribution using a quantum example oracle},\ }\href@noop {} {\bibfield  {journal} {\bibinfo  {journal} {SIAM Journal on Computing}\ }\textbf {\bibinfo {volume} {28}},\ \bibinfo {pages} {1136} (\bibinfo {year} {1999})},\ \bibinfo {note} {earlier version in COLT’95}\BibitemShut {NoStop}%
\bibitem [{\citenamefont {Arunachalam}\ and\ \citenamefont {de~Wolf}(2018)}]{Arunachalam2018}%
  \BibitemOpen
  \bibfield  {author} {\bibinfo {author} {\bibfnamefont {S.}~\bibnamefont {Arunachalam}}\ and\ \bibinfo {author} {\bibfnamefont {R.}~\bibnamefont {de~Wolf}},\ }\bibfield  {title} {\bibinfo {title} {Optimal quantum sample complexity of learning algorithms},\ }\href {http://jmlr.org/papers/v19/18-195.html} {\bibfield  {journal} {\bibinfo  {journal} {Journal of Machine Learning Research}\ }\textbf {\bibinfo {volume} {19}},\ \bibinfo {pages} {1} (\bibinfo {year} {2018})}\BibitemShut {NoStop}%
\bibitem [{\citenamefont {Vapnik}(1998)}]{vapnik_statistical_1998}%
  \BibitemOpen
  \bibfield  {author} {\bibinfo {author} {\bibfnamefont {V.~N.}\ \bibnamefont {Vapnik}},\ }\href@noop {} {\emph {\bibinfo {title} {Statistical Learning Theory}}}\ (\bibinfo  {publisher} {Wiley-Interscience},\ \bibinfo {year} {1998})\BibitemShut {NoStop}%
\bibitem [{\citenamefont {Anthony}\ and\ \citenamefont {Bartlett}(2009)}]{anthony2009neural}%
  \BibitemOpen
  \bibfield  {author} {\bibinfo {author} {\bibfnamefont {M.}~\bibnamefont {Anthony}}\ and\ \bibinfo {author} {\bibfnamefont {P.~L.}\ \bibnamefont {Bartlett}},\ }\href@noop {} {\emph {\bibinfo {title} {Neural network learning: Theoretical foundations}}}\ (\bibinfo  {publisher} {Cambridge University Press},\ \bibinfo {year} {2009})\BibitemShut {NoStop}%
\bibitem [{\citenamefont {Bădescu}\ and\ \citenamefont {O’Donnell}(2024)}]{B_descu_2024}%
  \BibitemOpen
  \bibfield  {author} {\bibinfo {author} {\bibfnamefont {C.}~\bibnamefont {Bădescu}}\ and\ \bibinfo {author} {\bibfnamefont {R.}~\bibnamefont {O’Donnell}},\ }\bibfield  {title} {\bibinfo {title} {Improved quantum data analysis},\ }\bibfield  {journal} {\bibinfo  {journal} {TheoretiCS}\ }\textbf {\bibinfo {volume} {Volume 3}},\ \href {https://doi.org/10.46298/theoretics.24.7} {10.46298/theoretics.24.7} (\bibinfo {year} {2024})\BibitemShut {NoStop}%
\bibitem [{\citenamefont {Wolf}(2023)}]{wolf2023mathematical}%
  \BibitemOpen
  \bibfield  {author} {\bibinfo {author} {\bibfnamefont {M.~M.}\ \bibnamefont {Wolf}},\ }\href@noop {} {\bibinfo {title} {Mathematical foundations of supervised learning}} (\bibinfo {year} {2023})\BibitemShut {NoStop}%
\bibitem [{\citenamefont {Caro}\ \emph {et~al.}(2021)\citenamefont {Caro}, \citenamefont {Gil-Fuster}, \citenamefont {Meyer}, \citenamefont {Eisert},\ and\ \citenamefont {Sweke}}]{caro2021encodingdependent}%
  \BibitemOpen
  \bibfield  {author} {\bibinfo {author} {\bibfnamefont {M.~C.}\ \bibnamefont {Caro}}, \bibinfo {author} {\bibfnamefont {E.}~\bibnamefont {Gil-Fuster}}, \bibinfo {author} {\bibfnamefont {J.~J.}\ \bibnamefont {Meyer}}, \bibinfo {author} {\bibfnamefont {J.}~\bibnamefont {Eisert}},\ and\ \bibinfo {author} {\bibfnamefont {R.}~\bibnamefont {Sweke}},\ }\bibfield  {title} {\bibinfo {title} {Encoding-dependent generalization bounds for parametrized quantum circuits},\ }\href {https://doi.org/10.22331/q-2021-11-17-582} {\bibfield  {journal} {\bibinfo  {journal} {Quantum}\ }\textbf {\bibinfo {volume} {5}},\ \bibinfo {pages} {582} (\bibinfo {year} {2021})}\BibitemShut {NoStop}%
\bibitem [{\citenamefont {Esscher}(1932)}]{esscher1932probability}%
  \BibitemOpen
  \bibfield  {author} {\bibinfo {author} {\bibfnamefont {F.}~\bibnamefont {Esscher}},\ }\bibfield  {title} {\bibinfo {title} {On the probability function in the collective theory of risk},\ }\href@noop {} {\bibfield  {journal} {\bibinfo  {journal} {Skandinavisk Aktuarietidskrift}\ }\textbf {\bibinfo {volume} {15}},\ \bibinfo {pages} {175} (\bibinfo {year} {1932})}\BibitemShut {NoStop}%
\bibitem [{\citenamefont {Gerber}\ and\ \citenamefont {Shiu}(1994)}]{gerber1994option}%
  \BibitemOpen
  \bibfield  {author} {\bibinfo {author} {\bibfnamefont {H.~U.}\ \bibnamefont {Gerber}}\ and\ \bibinfo {author} {\bibfnamefont {E.~S.~W.}\ \bibnamefont {Shiu}},\ }\bibfield  {title} {\bibinfo {title} {Option pricing by esscher transforms},\ }\href {https://www.soa.org/globalassets/assets/library/research/transactions-of-society-of-actuaries/1990-99/1994/january/tsa94v46pt1-5.pdf} {\bibfield  {journal} {\bibinfo  {journal} {Transactions of the Society of Actuaries}\ }\textbf {\bibinfo {volume} {46}},\ \bibinfo {pages} {99} (\bibinfo {year} {1994})}\BibitemShut {NoStop}%
\bibitem [{\citenamefont {Qiu}\ \emph {et~al.}(2024)\citenamefont {Qiu}, \citenamefont {Koor},\ and\ \citenamefont {Rebentrost}}]{qiu_quantum_2024}%
  \BibitemOpen
  \bibfield  {author} {\bibinfo {author} {\bibfnamefont {Y.}~\bibnamefont {Qiu}}, \bibinfo {author} {\bibfnamefont {K.}~\bibnamefont {Koor}},\ and\ \bibinfo {author} {\bibfnamefont {P.}~\bibnamefont {Rebentrost}},\ }\bibfield  {title} {\bibinfo {title} {The quantum esscher transform},\ }\href@noop {} {\bibfield  {journal} {\bibinfo  {journal} {arXiv preprint arXiv:2401.07561}\ } (\bibinfo {year} {2024})}\BibitemShut {NoStop}%
\bibitem [{\citenamefont {Bakshi}\ \emph {et~al.}(2023)\citenamefont {Bakshi}, \citenamefont {Liu}, \citenamefont {Moitra},\ and\ \citenamefont {Tang}}]{bakshi2023learningquantumhamiltonianstemperature}%
  \BibitemOpen
  \bibfield  {author} {\bibinfo {author} {\bibfnamefont {A.}~\bibnamefont {Bakshi}}, \bibinfo {author} {\bibfnamefont {A.}~\bibnamefont {Liu}}, \bibinfo {author} {\bibfnamefont {A.}~\bibnamefont {Moitra}},\ and\ \bibinfo {author} {\bibfnamefont {E.}~\bibnamefont {Tang}},\ }\href {https://arxiv.org/abs/2310.02243} {\bibinfo {title} {Learning quantum hamiltonians at any temperature in polynomial time}} (\bibinfo {year} {2023}),\ \Eprint {https://arxiv.org/abs/2310.02243} {arXiv:2310.02243 [quant-ph]} \BibitemShut {NoStop}%
\bibitem [{\citenamefont {Anshu}\ \emph {et~al.}(2021)\citenamefont {Anshu}, \citenamefont {Arunachalam}, \citenamefont {Kuwahara},\ and\ \citenamefont {Soleimanifar}}]{Anshu_2021}%
  \BibitemOpen
  \bibfield  {author} {\bibinfo {author} {\bibfnamefont {A.}~\bibnamefont {Anshu}}, \bibinfo {author} {\bibfnamefont {S.}~\bibnamefont {Arunachalam}}, \bibinfo {author} {\bibfnamefont {T.}~\bibnamefont {Kuwahara}},\ and\ \bibinfo {author} {\bibfnamefont {M.}~\bibnamefont {Soleimanifar}},\ }\bibfield  {title} {\bibinfo {title} {Sample-efficient learning of interacting quantum systems},\ }\href {https://doi.org/10.1038/s41567-021-01232-0} {\bibfield  {journal} {\bibinfo  {journal} {Nature Physics}\ }\textbf {\bibinfo {volume} {17}},\ \bibinfo {pages} {931–935} (\bibinfo {year} {2021})}\BibitemShut {NoStop}%
\bibitem [{\citenamefont {Gu}\ \emph {et~al.}(2024)\citenamefont {Gu}, \citenamefont {Cincio},\ and\ \citenamefont {Coles}}]{gu2024practical}%
  \BibitemOpen
  \bibfield  {author} {\bibinfo {author} {\bibfnamefont {A.}~\bibnamefont {Gu}}, \bibinfo {author} {\bibfnamefont {L.}~\bibnamefont {Cincio}},\ and\ \bibinfo {author} {\bibfnamefont {P.~J.}\ \bibnamefont {Coles}},\ }\bibfield  {title} {\bibinfo {title} {Practical hamiltonian learning with unitary dynamics and gibbs states},\ }\href {https://doi.org/10.1038/s41467-024-02311-4} {\bibfield  {journal} {\bibinfo  {journal} {Nature Communications}\ }\textbf {\bibinfo {volume} {15}},\ \bibinfo {pages} {312} (\bibinfo {year} {2024})}\BibitemShut {NoStop}%
\bibitem [{\citenamefont {Yu}\ \emph {et~al.}(2023)\citenamefont {Yu}, \citenamefont {Sun}, \citenamefont {Han},\ and\ \citenamefont {Yuan}}]{Yu_2023}%
  \BibitemOpen
  \bibfield  {author} {\bibinfo {author} {\bibfnamefont {W.}~\bibnamefont {Yu}}, \bibinfo {author} {\bibfnamefont {J.}~\bibnamefont {Sun}}, \bibinfo {author} {\bibfnamefont {Z.}~\bibnamefont {Han}},\ and\ \bibinfo {author} {\bibfnamefont {X.}~\bibnamefont {Yuan}},\ }\bibfield  {title} {\bibinfo {title} {Robust and efficient hamiltonian learning},\ }\href {https://doi.org/10.22331/q-2023-06-29-1045} {\bibfield  {journal} {\bibinfo  {journal} {Quantum}\ }\textbf {\bibinfo {volume} {7}},\ \bibinfo {pages} {1045} (\bibinfo {year} {2023})}\BibitemShut {NoStop}%
\bibitem [{\citenamefont {Chen}\ and\ \citenamefont {de~Wolf}(2022)}]{chen2022quantumalgorithmslowerbounds}%
  \BibitemOpen
  \bibfield  {author} {\bibinfo {author} {\bibfnamefont {Y.}~\bibnamefont {Chen}}\ and\ \bibinfo {author} {\bibfnamefont {R.}~\bibnamefont {de~Wolf}},\ }\href {https://arxiv.org/abs/2110.13086} {\bibinfo {title} {Quantum algorithms and lower bounds for linear regression with norm constraints}} (\bibinfo {year} {2022}),\ \Eprint {https://arxiv.org/abs/2110.13086} {arXiv:2110.13086 [quant-ph]} \BibitemShut {NoStop}%
\bibitem [{\citenamefont {Zhu}\ \emph {et~al.}(2024)\citenamefont {Zhu}, \citenamefont {Lukens},\ and\ \citenamefont {Kirby}}]{Zhu2024connectionbetween}%
  \BibitemOpen
  \bibfield  {author} {\bibinfo {author} {\bibfnamefont {Z.}~\bibnamefont {Zhu}}, \bibinfo {author} {\bibfnamefont {J.~M.}\ \bibnamefont {Lukens}},\ and\ \bibinfo {author} {\bibfnamefont {B.~T.}\ \bibnamefont {Kirby}},\ }\bibfield  {title} {\bibinfo {title} {On the connection between least squares, regularization, and classical shadows},\ }\href {https://doi.org/10.22331/q-2024-08-29-1455} {\bibfield  {journal} {\bibinfo  {journal} {{Quantum}}\ }\textbf {\bibinfo {volume} {8}},\ \bibinfo {pages} {1455} (\bibinfo {year} {2024})}\BibitemShut {NoStop}%
\bibitem [{\citenamefont {Chakraborty}\ \emph {et~al.}(2023)\citenamefont {Chakraborty}, \citenamefont {Morolia},\ and\ \citenamefont {Peduri}}]{Chakraborty_2023}%
  \BibitemOpen
  \bibfield  {author} {\bibinfo {author} {\bibfnamefont {S.}~\bibnamefont {Chakraborty}}, \bibinfo {author} {\bibfnamefont {A.}~\bibnamefont {Morolia}},\ and\ \bibinfo {author} {\bibfnamefont {A.}~\bibnamefont {Peduri}},\ }\bibfield  {title} {\bibinfo {title} {Quantum regularized least squares},\ }\href {https://doi.org/10.22331/q-2023-04-27-988} {\bibfield  {journal} {\bibinfo  {journal} {Quantum}\ }\textbf {\bibinfo {volume} {7}},\ \bibinfo {pages} {988} (\bibinfo {year} {2023})}\BibitemShut {NoStop}%
\bibitem [{\citenamefont {Aminian}\ \emph {et~al.}(2025)\citenamefont {Aminian}, \citenamefont {Asadi}, \citenamefont {Li}, \citenamefont {Beirami}, \citenamefont {Reinert},\ and\ \citenamefont {Cohen}}]{aminian2025generalization}%
  \BibitemOpen
  \bibfield  {author} {\bibinfo {author} {\bibfnamefont {G.}~\bibnamefont {Aminian}}, \bibinfo {author} {\bibfnamefont {A.~R.}\ \bibnamefont {Asadi}}, \bibinfo {author} {\bibfnamefont {T.}~\bibnamefont {Li}}, \bibinfo {author} {\bibfnamefont {A.}~\bibnamefont {Beirami}}, \bibinfo {author} {\bibfnamefont {G.}~\bibnamefont {Reinert}},\ and\ \bibinfo {author} {\bibfnamefont {S.~N.}\ \bibnamefont {Cohen}},\ }\bibfield  {title} {\bibinfo {title} {Generalization and robustness of the tilted empirical risk},\ }in\ \href {https://openreview.net/forum?id=Wvc6d6926j} {\emph {\bibinfo {booktitle} {Forty-second International Conference on Machine Learning}}}\ (\bibinfo {year} {2025})\BibitemShut {NoStop}%
\bibitem [{\citenamefont {Mitarai}\ \emph {et~al.}(2018)\citenamefont {Mitarai}, \citenamefont {Negoro}, \citenamefont {Kitagawa},\ and\ \citenamefont {Fujii}}]{mitarai2018quantum}%
  \BibitemOpen
  \bibfield  {author} {\bibinfo {author} {\bibfnamefont {K.}~\bibnamefont {Mitarai}}, \bibinfo {author} {\bibfnamefont {M.}~\bibnamefont {Negoro}}, \bibinfo {author} {\bibfnamefont {M.}~\bibnamefont {Kitagawa}},\ and\ \bibinfo {author} {\bibfnamefont {K.}~\bibnamefont {Fujii}},\ }\bibfield  {title} {\bibinfo {title} {Quantum circuit learning},\ }\href {https://doi.org/10.1103/PhysRevA.98.032309} {\bibfield  {journal} {\bibinfo  {journal} {Phys. Rev. A}\ }\textbf {\bibinfo {volume} {98}},\ \bibinfo {pages} {032309} (\bibinfo {year} {2018})}\BibitemShut {NoStop}%
\bibitem [{\citenamefont {Schuld}\ \emph {et~al.}(2019)\citenamefont {Schuld}, \citenamefont {Bergholm}, \citenamefont {Gogolin}, \citenamefont {Izaac},\ and\ \citenamefont {Killoran}}]{Schuld_2019}%
  \BibitemOpen
  \bibfield  {author} {\bibinfo {author} {\bibfnamefont {M.}~\bibnamefont {Schuld}}, \bibinfo {author} {\bibfnamefont {V.}~\bibnamefont {Bergholm}}, \bibinfo {author} {\bibfnamefont {C.}~\bibnamefont {Gogolin}}, \bibinfo {author} {\bibfnamefont {J.}~\bibnamefont {Izaac}},\ and\ \bibinfo {author} {\bibfnamefont {N.}~\bibnamefont {Killoran}},\ }\bibfield  {title} {\bibinfo {title} {Evaluating analytic gradients on quantum hardware},\ }\bibfield  {journal} {\bibinfo  {journal} {Physical Review A}\ }\textbf {\bibinfo {volume} {99}},\ \href {https://doi.org/10.1103/physreva.99.032331} {10.1103/physreva.99.032331} (\bibinfo {year} {2019})\BibitemShut {NoStop}%
\bibitem [{\citenamefont {Wierichs}\ \emph {et~al.}(2022)\citenamefont {Wierichs}, \citenamefont {Izaac}, \citenamefont {Wang},\ and\ \citenamefont {Lin}}]{wierichs2022general}%
  \BibitemOpen
  \bibfield  {author} {\bibinfo {author} {\bibfnamefont {D.}~\bibnamefont {Wierichs}}, \bibinfo {author} {\bibfnamefont {J.}~\bibnamefont {Izaac}}, \bibinfo {author} {\bibfnamefont {C.}~\bibnamefont {Wang}},\ and\ \bibinfo {author} {\bibfnamefont {C.~Y.-Y.}\ \bibnamefont {Lin}},\ }\bibfield  {title} {\bibinfo {title} {General parameter-shift rules for quantum gradients},\ }\href@noop {} {\bibfield  {journal} {\bibinfo  {journal} {Quantum}\ }\textbf {\bibinfo {volume} {6}},\ \bibinfo {pages} {677} (\bibinfo {year} {2022})}\BibitemShut {NoStop}%
\bibitem [{\citenamefont {Rumelhart}\ \emph {et~al.}(1986)\citenamefont {Rumelhart}, \citenamefont {Hinton},\ and\ \citenamefont {Williams}}]{Rumelhart1986LearningRB}%
  \BibitemOpen
  \bibfield  {author} {\bibinfo {author} {\bibfnamefont {D.~E.}\ \bibnamefont {Rumelhart}}, \bibinfo {author} {\bibfnamefont {G.~E.}\ \bibnamefont {Hinton}},\ and\ \bibinfo {author} {\bibfnamefont {R.~J.}\ \bibnamefont {Williams}},\ }\bibfield  {title} {\bibinfo {title} {Learning representations by back-propagating errors},\ }\href@noop {} {\bibfield  {journal} {\bibinfo  {journal} {Nature}\ }\textbf {\bibinfo {volume} {323}},\ \bibinfo {pages} {533} (\bibinfo {year} {1986})}\BibitemShut {NoStop}%
\bibitem [{\citenamefont {Abbas}\ \emph {et~al.}(2023)\citenamefont {Abbas}, \citenamefont {King}, \citenamefont {Huang}, \citenamefont {Huggins}, \citenamefont {Movassagh}, \citenamefont {Gilboa},\ and\ \citenamefont {McClean}}]{abbas2023quantumbackpropagationinformationreuse}%
  \BibitemOpen
  \bibfield  {author} {\bibinfo {author} {\bibfnamefont {A.}~\bibnamefont {Abbas}}, \bibinfo {author} {\bibfnamefont {R.}~\bibnamefont {King}}, \bibinfo {author} {\bibfnamefont {H.-Y.}\ \bibnamefont {Huang}}, \bibinfo {author} {\bibfnamefont {W.~J.}\ \bibnamefont {Huggins}}, \bibinfo {author} {\bibfnamefont {R.}~\bibnamefont {Movassagh}}, \bibinfo {author} {\bibfnamefont {D.}~\bibnamefont {Gilboa}},\ and\ \bibinfo {author} {\bibfnamefont {J.~R.}\ \bibnamefont {McClean}},\ }\href {https://arxiv.org/abs/2305.13362} {\bibinfo {title} {On quantum backpropagation, information reuse, and cheating measurement collapse}} (\bibinfo {year} {2023}),\ \Eprint {https://arxiv.org/abs/2305.13362} {arXiv:2305.13362 [quant-ph]} \BibitemShut {NoStop}%
\bibitem [{\citenamefont {Hur}\ and\ \citenamefont {Park}(2024)}]{hur2024understandinggeneralizationquantummachine}%
  \BibitemOpen
  \bibfield  {author} {\bibinfo {author} {\bibfnamefont {T.}~\bibnamefont {Hur}}\ and\ \bibinfo {author} {\bibfnamefont {D.~K.}\ \bibnamefont {Park}},\ }\href {https://arxiv.org/abs/2411.06919} {\bibinfo {title} {Understanding generalization in quantum machine learning with margins}} (\bibinfo {year} {2024}),\ \Eprint {https://arxiv.org/abs/2411.06919} {arXiv:2411.06919 [quant-ph]} \BibitemShut {NoStop}%
\bibitem [{\citenamefont {Neyshabur}\ \emph {et~al.}(2018)\citenamefont {Neyshabur}, \citenamefont {Bhojanapalli},\ and\ \citenamefont {Srebro}}]{neyshabur2018pacbayesianapproachspectrallynormalizedmargin}%
  \BibitemOpen
  \bibfield  {author} {\bibinfo {author} {\bibfnamefont {B.}~\bibnamefont {Neyshabur}}, \bibinfo {author} {\bibfnamefont {S.}~\bibnamefont {Bhojanapalli}},\ and\ \bibinfo {author} {\bibfnamefont {N.}~\bibnamefont {Srebro}},\ }\href {https://arxiv.org/abs/1707.09564} {\bibinfo {title} {A pac-bayesian approach to spectrally-normalized margin bounds for neural networks}} (\bibinfo {year} {2018}),\ \Eprint {https://arxiv.org/abs/1707.09564} {arXiv:1707.09564 [cs.LG]} \BibitemShut {NoStop}%
\bibitem [{\citenamefont {Hanneke}\ and\ \citenamefont {Kontorovich}(2020)}]{hanneke2020stablesamplecompressionschemes}%
  \BibitemOpen
  \bibfield  {author} {\bibinfo {author} {\bibfnamefont {S.}~\bibnamefont {Hanneke}}\ and\ \bibinfo {author} {\bibfnamefont {A.}~\bibnamefont {Kontorovich}},\ }\href {https://arxiv.org/abs/2011.04586} {\bibinfo {title} {Stable sample compression schemes: New applications and an optimal svm margin bound}} (\bibinfo {year} {2020}),\ \Eprint {https://arxiv.org/abs/2011.04586} {arXiv:2011.04586 [cs.LG]} \BibitemShut {NoStop}%
\bibitem [{\citenamefont {Bousquet}\ and\ \citenamefont {Elisseeff}(2002)}]{bousquet2002stability}%
  \BibitemOpen
  \bibfield  {author} {\bibinfo {author} {\bibfnamefont {O.}~\bibnamefont {Bousquet}}\ and\ \bibinfo {author} {\bibfnamefont {A.}~\bibnamefont {Elisseeff}},\ }\bibfield  {title} {\bibinfo {title} {Stability and generalization},\ }\href {https://doi.org/10.1162/153244302760200704} {\bibfield  {journal} {\bibinfo  {journal} {J. Mach. Learn. Res.}\ }\textbf {\bibinfo {volume} {2}},\ \bibinfo {pages} {499–526} (\bibinfo {year} {2002})}\BibitemShut {NoStop}%
\end{thebibliography}%

\end{document}